%% file: main.tex
\documentclass[10pt,journal]{IEEEtran}
\input{defspack}
\begin{document}
\title{Performance-Complexity Tradeoffs in Greedy Weak Submodular Maximization with Random Sampling}
\author{Abolfazl~Hashemi$^\dagger$, Haris~Vikalo, and Gustavo de Veciana
\thanks{
	Abolfazl Hashemi is with the School of Electrical and Computer Engineering, Purdue University. Haris Vikalo and Gustavo de Veciana are with the Department of Electrical and Computer Engineering, University of Texas at Austin. $^\dagger$Work done while at the Department of Electrical and Computer Engineering, The University of Texas at Austin.
This work was supported in part by NSF grant ECCS-1809327. Part of the results in this paper were presented at the IEEE International Conference on Acoustics, Speech, and Signal
Processing,Toronto, Ontario, Canada, June 2021 \cite{hashemi2021performance}.}}
\maketitle
\begin{abstract}
Many problems in signal processing and machine learning can be formalized as weak submodular optimization tasks. For such problems, a simple greedy algorithm (\textsc{Greedy}) is guaranteed to find a solution achieving the objective with a value no worse than $1-e^{-1/c}$ of the optimal, where $c$ is the multiplicative weak-submodularity constant. Due to the high cost of querying large-scale systems, the complexity of \textsc{Greedy} becomes prohibitive in contemporary applications. In this work, we study the tradeoff between performance and complexity when one resorts to random sampling strategies to reduce the query complexity of \textsc{Greedy}.  Specifically, we quantify the effect of uniform sampling strategies on \textsc{Greedy}'s performance through two metrics: (i) probability of identifying an optimal subset, and (ii) suboptimality with respect to the optimal solution. The latter implies that uniform sampling strategies with a fixed sampling size achieve a non-trivial approximation factor; however, we show that with overwhelming probability, these methods fail to find the optimal subset. Our analysis shows that the failure of uniform sampling strategies with fixed sample size can be circumvented by successively increasing the size of the search space. Building upon this insight, we propose a simple progressive stochastic greedy algorithm and study its approximation guarantees. Moreover, we demonstrate effectiveness of the proposed method in dimensionality reduction applications and feature selection tasks for clustering and object tracking.
\end{abstract}
\begin{IEEEkeywords}
weak submodular optimization, greedy algorithms, randomized algorithms, subset selection,
\end{IEEEkeywords}
\section{Introduction}\label{sec:intro}
\input{Secintro}
\section{Background}\label{sec:back}
\input{Secback}

\section{Analysis of Success Probability}\label{sec:fail}
\input{Secpre}
\section{Analysis of Approximation Factor}\label{sec:alg}
\input{Secalg}
\section{Applications to Subset Selection}\label{sec:sim}
\input{Secsim}
\section{Conclusion} \label{sec:conc}
In this paper, we studied the problem of large-scale monotone weak submodular maximization that comes up in many modern signal processing and machine learning applications including sparse reconstruction, dimensionality reduction, observation gathering, and sensor selection. Motivated by the desire to reduce complexity of the celebrated greedy scheme, we theoretically studied fundamental performance limits of restricting the size of the greedy search space by means of uniform sampling strategies. We first studied the asymptotic probability of successfully identifying the optimal subset. Our analysis revealed that many of the standard practices that rely on fixed sampling sizes lead to a success probability that is asymptotically zero.  We showed that an increasing schedule of the search space size satisfies a necessary condition for the exact identification of the optimal subset in large-scale problems. Following this insight, we proposed a progressive stochastic greedy algorithm and demonstrated its efficacy in the applications to sparse subset selection, dimensionality reduction, and extended object tracking. We further established strong guarantees, both on expectations and with high probability, on the approximation factors of the proposed algorithm.
	
Our established framework gives rise to interesting open problems. First, the proposed analysis can be employed to study exact identification conditions of \textsc{Psg} in other classes of weak submodular maximization problems, i.e., settings beyond the sparse support selection task considered in this paper. For instance, exact identification conditions of \textsc{Greedy} for the task of observation selection were recently considered by  \cite{sharma2015greedy}. Utilizing our framework, similar results may be established for \textsc{Psg}. Furthermore, as we argued, the intersection of search space of \textsc{Psg} and new elements from the optimal subset is nonempty with high probability  for all iterations. Since \textsc{Greedy} considers all the elements of the ground set in each iteration, the intersection of the search space of \textsc{Greedy} and new elements from the optimal subset is always nonempty. Therefore, it is reasonable to ask whether the set of conditions under which \textsc{Greedy} and \textsc{Psg} exactly identify the optimal subset of a given weak submodular optimization problem $\Ps$ is the same. Indeed, this is the case for the problem of sparse support selection. Finally, studying general weak submodular optimization problems remains of interest.
\section*{Acknowledgment}
This work was supported in part by NSF grant ECCS-1809327.
\begin{appendices}
\section{Auxiliary Lemma}
\begin{lemma} \label{lem:ineq}
For every $|a|\leq 1$ and $b\geq 1$ it holds that $(1+a)^b \geq e^{ab}(1-a^2b)$.
\end{lemma}
\begin{proof}
	Let x = $ab$, $|x|\leq b$. 
	Consider $g(x) = e^{-x}(1+\frac{x}{b})^b-(1-\frac{x^2}{b})$. At $x = 0$, both $g(x)$ and $f^\prime(x)$ are zero. If $f^\prime(x) = 0$ for any other $x$ in the interval, for such $x$ we have
	$$e^{-x}(1+\frac{x}{b})^b = 2+\frac{2x}{b}.$$
	Therefore, for such $x$ 
	$$g(x) = \frac{(x+1)^2}{b}+1-\frac{1}{b}>0.$$
	Furthermore, since $g(b)>0$ for all $b$ while $g(-b)>0$ for $b>1$ and $g(-b)=0$ for $b=1$, all other points we must have $g(x)>0$.
\end{proof}
\section{Azuma’s inequality \cite{chung2006concentration}}
\begin{theorem}\label{thm:mar}
Let $X$ be a martingale associated with a filtration $\mathcal{F}$ and a sequence of random variables $X_0,X_1,\dots,X_n$ satisfying $X_i = \E[X | \mathcal{F}_i]$ and, in particular, $X_0 = \E[X]$ and $X_n = X$. If $|X_i-X_{i-1}|\leq c_i$, then, 
\begin{equation}
    \Pr\left(X-\E[X]>\lambda\right) \leq \exp\left(-\frac{\lambda^2}{2\sum_{i=1}^nc_i^2}\right).
\end{equation}
\end{theorem}
\end{appendices}
\bibliographystyle{ieeetr}
\bibliography{refs}
\section*{Supplementary Material}
In this supplementary document, we demonstrate an application of our results to the problem of sparse support selection.

The goal of sparse reconstruction, or sparse support selection, is to reconstruct a sparse vector from a relatively small 
number of its linear measurements. In particular, we are given a linear measurement model
\begin{equation} \label{linm}
	{\y=\A\x+\e},
\end{equation}
where $\x \in\Rb^{m}$ is a $k$-sparse unknown vector, i.e., a vector with at most $k$ non-zero components, $\y \in\Rb^{n}$ denotes the vector of measurements, $\A \in\Rb^{n \times m}$ is the coefficient matrix assumed to be full rank, and ${\bf \e} \in\R^{n}$ denotes the additive measurement noise vector. For simplicity, we here focus on the case $\e = \mathbf{0}$ and $\A \sim \N\left(0,\frac{1}{n}\right)$. 
The search for a sparse approximation of $\x$ leads to the NP-hard cardinality-constrained least-squares problem
\begin{equation}  \label{sls}
	\begin{aligned}
		& \underset{\x}{\text{minimize}}
		\quad\|\y-\A\x\|^{2}_{2}
		\qquad \text{subject to} \quad
		\|\x\|_{0} \leq k.
	\end{aligned}
\end{equation}
One can readily reformulate \eqref{sls} as a subset selection task according to the following procedure. For a fixed subset 
$\S \subset [m]$ where $|\S|\leq n$, we can find an approximation to $\x$ via the least-squares solution
$\x_{LS} = \A_\S^\dagger \y$, where $\A_\S^\dagger=\left(\A_\S^{\top}\A_\S\right)^{-1}\A_\S^{\top}$ denotes the 
Moore-Penrose pseudo-inverse of $\A_\S$. Finding the optimal $k$-sparse vector $\x^\star$ is equivalent to identifying the 
support of $\x^\star$, i.e., determining the set of nonzero entries of $\x^\star$ which we denote by $\S^\star$. More formally,
\eqref{sls} is recast as
\begin{equation}  \label{sls-set}
	\begin{aligned}
		& \underset{\S}{\text{minimize}}
		\quad\|\y-\P(\S)\y\|^{2}_{2}
		\qquad \text{subject to} \quad
		|\S| \leq k,
	\end{aligned}
\end{equation}
where $\P(\S) = \A_\S\A_\S^\dagger$ is the projection operator onto the subspace spanned by the columns of $\A_\S$. 
Since $\|\y\|_2^2 = \|\y-\P(\S)\y\|^{2}_{2}+\|\P(\S)\y\|^{2}_{2}$, \eqref{sls-set} can equivalently be written as
\begin{equation}  \label{sls-sub}
	\begin{aligned}
		& \underset{\S}{\text{maximize}}
		\quad g(\S) := \|\P(\S)\y\|^{2}_{2}
		\qquad \text{subject to} \quad 
		|\S| \leq k.
	\end{aligned}
\end{equation}
Note that since $\A$ is full rank, it can be shown that \eqref{sls-sub} has a unique 
solution.

Theorem \ref{thm:sss} summarizes the  result of this supplementary. It states that \textsc{Psg} successfully recovers 
$k$-sparse $\x$ with high probability as long as the number of measurements is linear in $k$ (sparsity) and logarithmic in $\frac{m}{k}$, achieving 
the optimal sample complexity established by Candes and Tao \cite{candes2005decoding}.
\begin{theorem}\label{thm:sss}
	Let $\x \in \Rb^m$ be an arbitrary sparse vector with $k$ non-zero entries  and let $\A \in \Rb^{n\times m}$ denote a random matrix with entries drawn independently from ${\cal N}(0,1/n)$. Given noiseless measurements $\y=\A\x$, \textsc{Psg} with parameter $e^{-k}\leq \epsilon \leq e^{-\frac{k}{m}}$ finds a solution that satisfies
	\begin{equation}
		\Pr\left(\S_{psg} = \S^\star\right)\geq \left(1-\epsilon\right)^{k-\log\frac{1}{\epsilon}} 
		\left(1-c_1(\frac{m}{k})^{c_2}\exp(-c_3\frac{n}{k})\right),
	\end{equation}
	for some positive universal constants $c_1$, $c_2$, and $c_3$.
	Furthermore, assume that $m> k\sqrt{k}$ and
	\begin{equation}\label{eq:ssamcomp}
		n \geq \max\left(\frac{6}{C_1} k \log \frac{m}{k\sqrt[6]{4\beta}}\text{ }, \text{ }C_2k\right),
	\end{equation}
	where $0<\beta<1$, and $C_1$ and $C_2$ are positive constants independent of $\beta$, $n$, $m$, and 
	$k$. Then, \textsc{Psg} with parameter $\epsilon < \frac{\beta}{k}$ can exactly identify the optimal support subset $\S^\star$ with a probability of success exceeding $1-2\beta$.
\end{theorem}
\begin{proof}[Proof of Theorem \ref{thm:sss}]
	To prove the theorem it suffices to derive nontrivial  lower bounds on $\prod_{i = 0}^{k-1} p_{psg}^{(i)}$ and $\prod_{i = 0}^{k-1} q_{psg}^{(i)}$ (as per Lemma 1 in the main paper).
	
	Theorem 3 in the main paper provides a lower bound on  $\prod_{i=0}^{k-1} p_{psg}^{(i)}$.
	Therefore, it just remains to derive a nontrivial lower bound on $q_{psg}^{(i)}$ in order to show existence of a sufficient condition for the exact identification of $\S^\star$ and establish a lower bound on the probability of success of \textsc{Psg}. A lower bound on $q_{psg}^{(i)}$ can be obtained by considering the conditions under which the largest marginal gain of  elements in $\R_{psg}^{(i)}\cap \S^\star$ exceeds that in $\R_{psg}^{(i)}\backslash \S^\star$ for all $i= 0,\dots,k-1$, that is, 
	\begin{equation}\label{eq:sofcon}
		\max_{j\in \R_{psg}^{(i)}\backslash \S^\star} g_j(\S_{psg}^{(i)}) < \max_{j\in \R_{psg}^{(i)}\cap \S^\star}g_j(\S_{psg}^{(i)}),
	\end{equation}
	with high probability. In Lemma \ref{lem:sss1}, we show the sufficient condition defined in \eqref{eq:sofcon} holds with high probability for \textsc{Psg} applied to the problem of sparse support selection.
	\begin{lemma}\label{lem:sss1}
		Let $\x \in \Rb^m$ be an arbitrary sparse vector with $k < m$ non-zero entries and let $\A \in \Rb^{n\times m}$ denote a random matrix with entries drawn independently from ${\cal N}(0,1/n)$. Given noiseless measurements $\y=\A\x$, for \textsc{Psg} with parameter $e^{-k}\leq \epsilon \leq e^{-\frac{k}{m}}$ it holds that $\prod_{i = 0}^{k-1}q_{psg}^{(i)} \geq \tilde{q}_1\tilde{q}_2$ where 
		\begin{equation} \label{eq:probnonois}
			\begin{aligned}
				\tilde{q}_1&=\left(1-2\exp\left(-n (\frac{\gamma^2}{4}-\frac{\gamma^3}{6})\right)\right)^m-\exp(-\delta^2\frac{n}{2}), \mbox{ and }\\
				\tilde{q}_2&=\left(1-\exp\left(-\frac{1-\gamma}{1+\gamma}(1-\sqrt{\frac{k}{n}}-\delta)^2\frac{n}{2k}\right)\right)^{k(m-k)},
			\end{aligned}
		\end{equation}
		for any $0<\gamma<1$ and $\delta>0$.
	\end{lemma}
	Before proving Lemma \ref{lem:sss1}, we state four lemmas that are used in the proof.
	
	Lemma \ref{lem:23} (Lemma 3.1 in \cite{hashemi2018accelerated}) states that the Euclidean norm of a normally distributed vector is  concentrated around its expected value.
	\begin{lemma} \label{lem:23}
		Let $\a \in \Rb^{n}$ be a vector consisting of entries that are drawn independently from ${\cal N}(0,1/n)$.
		Then it holds that $\E[\|\a\|_2^2]=1$. Furthermore, one can show that
		\begin{equation}
			\Pr\left(1-\gamma<\|\a\|_2^2<1+\gamma\right)\geq 1-2e^{-nc_0(\gamma)},
		\end{equation}
		where $c_0(\gamma)=\frac{\gamma^2}{4}-\frac{\gamma^3}{6}$ for $0<\gamma <1$.
	\end{lemma}
	Lemma \ref{lem:nois} (Corollary 2.4.5 in \cite{golub2012matrix}) states inequalities between the 
	maximum and minimum singular values of a matrix and its submatrices.
	\begin{lemma}\label{lem:nois} 
		Let $\C$ be a full rank tall matrix and let $\A$ be a submatrix of $\C$. Then
		\begin{equation} \label{eq:sigA}
			\begin{aligned} 
				&\sigma_{\min}\left(\A\right)\geq \sigma_{\min}\left(\C\right), &\sigma_{\max}\left(\A\right)\leq \sigma_{\max}\left(\C\right).
			\end{aligned}
		\end{equation}
	\end{lemma}
	Lemma \ref{lem:davidson} from \cite{davidson2001local} establishes a probabilistic bound on the smallest singular value of a normally distributed matrix.
	\begin{lemma}\label{lem:davidson}
		Let $\A \in \Rb^{n\times k}$ denote a tall matrix whose entries  
		are drawn independently from ${\cal N}(0,1/n)$.
		Then for any $\delta>0$ it holds that
		\begin{equation}
			\Pr(\sigma_{\min}\left(\A\right)\geq 1- \sqrt{\frac{k}{n}}-\delta)\geq 1-\exp(-\delta^2\frac{n}{2}),
		\end{equation}
		and
		\begin{equation}
			\Pr(\sigma_{\max}\left(\A\right)\geq 1+ \sqrt{\frac{k}{n}}+\delta)\geq 1-\exp(-\delta^2\frac{n}{2}).
		\end{equation}
	\end{lemma}
	Lemma \ref{lem:tropp} (Proposition 4 in \cite{tropp2007signal}) establishes an upper bound on the inner product of two independent random vectors.
	\begin{lemma}\label{lem:tropp}
		Let $\a \in \Rb^{n}$ denote a vector with entries that 
		are drawn independently from ${\cal N}(0,1/n)$. Let $\u\in \Rb^{n}$ be a random vector such that $\|\u\|_2=1$ and let $\u$ and $\a$ be statistically independent. Then for $\delta>0$ it holds
		\begin{equation}
			\Pr(|\a^\top\u|\leq \delta)\geq 1-\exp(-\delta^2\frac{n}{2}).
		\end{equation}
	\end{lemma}
	We are now ready to proceed with the proof of Lemma \ref{lem:sss1}. Let $\r_i := (\I_n-\P(\S^{(i)}_{psg}))\y$ be the {\it residual} vector in the $i\ts{th}$ iteration of \textsc{Psg}. Note that if in the previous iterations \textsc{Psg} selected columns of $\A$ with indices from $\S^\star$, the selected columns are orthogonal to $\r_i$.
	
	To prove the stated result it is sufficient to establish a lower bound on the probability of \eqref{eq:sofcon}. Given the selection criterion of \textsc{Omp} for sparse support selection, it is straightforward to see that
	\begin{equation}
		\rho(\r_i) :=\max_{j\in \R_{psg}^{(i)}\backslash \S^\star} \frac{|\a_j^\top\r_i|}{\|\a_j\|_2}\Big\slash \max_{j\in \R_{psg}^{(i)}\cap (\S^\star\backslash \S^{(i)}_{psg})} \frac{|\a_j^\top\r_i|}{\|\a_j\|_2}<1
	\end{equation}
	is a sufficient condition for successful identification of an element from $\R_{psg}^{(i)}\cap (\S^\star\backslash \S^{(i)}_{psg})$. Our goal in this theorem is to prove that with high probability $\rho(\r_i) <1$ in each iteration $i$. This in turn will establish a lower bound on $q_{psg}^{(i)}$, $i = 0,\dots,k-1$. To this end, following \cite{tropp2007signal} we employ an induction technique to show that $\rho(\r_i) <1$ if $\R_{psg}^{(i)} \cap (\S^\star\backslash\S_{psg}^{(i)}) \neq \emptyset$ and $\S_{psg}^{(i)} \subseteq \S^\star$. Since computing $\rho(\r_i)$ appears challenging, to establish the desired results we show that a judicious upper bound on $\rho(\r_i)$ is with overwhelming probability smaller than $1$. In particular, note that one may upper bound $\rho(\r_i)$ as
	\begin{equation}\label{eq:rho1}
		\begin{aligned}
			\rho(\r_i) &\leq \frac{\max_{j\in \R_{psg}^{(i)}\backslash \S^\star}|\a_j^\top\r_i|}{\max_{j\in \R_{psg}^{(i)}\cap (\S^\star\backslash \S^{(i)}_{psg})}|\a_j^\top\r_i|} &\\\cdot \frac{\max_{j\in \R_{psg}^{(i)}\cap (\S^\star\backslash \S^{(i)}_{psg})}\|\a_j\|_2}{\min_{j\in \R_{psg}^{(i)}\backslash \S^\star}\|\a_j\|_2},\\
			&\leq \frac{\max_{j\in \R_{psg}^{(i)}\backslash \S^\star}|\a_j^\top\r_i|}{\max_{j\in \R_{psg}^{(i)}\cap (\S^\star\backslash \S^{(i)}_{psg})}|\a_j^\top\r_i|}\cdot \frac{\max_{j\in [m]}\|\a_j\|_2}{\min_{j\in [m]}\|\a_j\|_2}.
		\end{aligned}
	\end{equation}
	Let $\Z_1$ denote the event that
	\begin{equation}\label{eq:Z1}
		\frac{\max_{j\in [m]}\|\a_j\|_2}{\min_{j\in [m]}\|\a_j\|_2} \leq \sqrt{\frac{1+\gamma}{1-\gamma}}
	\end{equation}
	for some $\gamma \in (0,1)$. Then, from Lemma \ref{lem:23} it follows that
	\begin{equation}
		\Pr(\Z_1)\geq\left(1-2e^{-c_0(\gamma)n}\right)^m.
	\end{equation}
	In other words, since $\|\a_j\|_2$'s are highly concentrated around one, one can approximate \eqref{eq:rho1} by disregarding the second factor on the right-hand side. Additionally, let $\Z_2$ denote the event that $\sigma_{\min}(\A_{\S^\star}) \geq 1-\sqrt{\frac{k}{n}}-\delta$ for some $\delta>0$. Then, from Lemma \ref{lem:davidson} we have
	\begin{equation}
		\Pr(\Z_2)\geq1-\exp(-\delta^2\frac{n}{2}).
	\end{equation}
	Therefore, by conditioning
	\begin{equation}
		\begin{aligned}
			\Pr\left(\rho(\r_i) <1\right) \geq \Pr\left(\rho(\r_i) <1\text{ }| \text{ }\Z_1\cap \Z_2\right) \Pr(\Z_1\cap \Z_2).
		\end{aligned}
	\end{equation}
	Note that occurrence of $\Z_1$ and $\Z_2$ in the $i = 0$ iteration implies $\Z_1$ and $\Z_1$ occur throughout the algorithm. Thus, $\Z_1$ and $\Z_2$ are in a sense {\it global} events. Note that $\Pr(\Z_1\cap \Z_2)$ can be bounded according to
	\begin{equation}
		\begin{aligned}
			\Pr(\Z_1\cap \Z_2) &= \Pr(\Z_1)+\Pr( \Z_2)-\Pr(\Z_1\cup \Z_2),\\
			&\geq \Pr(\Z_1)+\Pr( \Z_2)-1,\\
			&\geq\left(1-2e^{-c_0(\gamma)n}\right)^m-\exp(-\delta^2\frac{n}{2}):=\tilde{q}_1.
		\end{aligned}
	\end{equation}
	
	Now, note that $\max_{j\in \R_{psg}^{(i)}\cap (\S^\star\backslash \S^{(i)}_{psg})}|\a_j^\top\r_i|$ can alternatively be written as an $\ell_\infty$-norm of its argument. Furthermore, since $|\R_{psg}^{(i)}\cap (\S^\star\backslash \S^{(i)}_{psg})|\leq |\S^\star| \leq k$, there are at most $k$ inner products $|\a_j^\top\r_i|$ to consider (i.e., $1 \le j \le k$). Finally, since for a $k$-dimensional vector $\a$ holds that $\sqrt{k}\|\a\|_\infty \geq \|\a\|_2$, by conditioning on $\Z_1\cap \Z_2$ we have	
	\begin{equation}\label{eq:40app}
		\begin{aligned}
			\rho(\r_i) &\leq \sqrt{k}\frac{\max_{j\in \R_{psg}^{(i)}\backslash \S^\star}|\a_j^\top\r_i|}{\| \A_{\R_{psg}^{(i)}\cap (\S^\star\backslash \S^{(i)}_{psg})}^\top\r_i\|_2} \sqrt{\frac{1+\gamma}{1-\gamma}},\\
			&=\frac{\sqrt{k}}{c_1(\gamma)}\max_{j\in \R_{psg}^{(i)}\backslash \S^\star}|\a_j^\top\tilde{\r}_i|,
		\end{aligned}
	\end{equation}
	where $c_1(\gamma)=\sqrt{\frac{1-\gamma}{1+\gamma}}$ and $\tilde{\r}_i=\r_i\slash \| \A_{\R_{psg}^{(i)}\cap (\S^\star\backslash \S^{(i)}_{psg})}^\top\r_i\|_2$. Note that $\tilde{\r}_i$ is introduced in part to help us apply the concentration results established by Lemma \ref{lem:tropp}. Since $\A_{\R_{psg}^{(i)}\cap (\S^\star\backslash \S^{(i)}_{psg})}$ is a submatrix of $\A_{\S^\star}$, by conditioning on $\Z_1\cap \Z_2$, properties of singular values, and Lemma \ref{lem:nois} we obtain
	\begin{equation}\label{eq:Z2}
		\begin{aligned}
			\|\tilde{\r}_i\|_2 &= \frac{\|\r_i \|_2}{\| \A_{\R_{psg}^{(i)}\cap (\S^\star\backslash \S^{(i)}_{psg})}^\top\r_i\|_2}, \\
			&\leq \frac{1}{\sigma_{\min}(_{\R_{psg}^{(i)}\cap (\S^\star\backslash \S^{(i)}_{psg})})},\\ &\leq\frac{1}{\sigma_{\min}(\A_{\S^\star})}, \\
			&\leq \frac{1}{1-\sqrt{\frac{k}{n}}-\delta}.
		\end{aligned}
	\end{equation}
	By defining $\bar{\r}_i = \sigma_{\min}(\A_{\S^\star})\tilde{\r}_i$, $\|\bar{\r}_i\|_2 = 1$, conditioning on $\Z_1\cap \Z_2$ \eqref{eq:40app} can be written as
	\begin{equation}
		\begin{aligned}
			\rho(\r_i) &\leq\frac{\sqrt{k}}{c_1(\gamma)(1-\sqrt{\frac{k}{n}}-\delta)}\max_{j\in \R_{psg}^{(i)}\backslash \S^\star}|\a_j^\top\bar{\r}_i|\\ &\leq\frac{\sqrt{k}}{c_1(\gamma)(1-\sqrt{\frac{k}{n}}-\delta)}\max_{j\in [m]\backslash \S^\star}|\a_j^\top\bar{\r}_i|
		\end{aligned}
	\end{equation}
	Thus, conditioning on $\Z_1$ and $\Z_2$
	\begin{equation}
		\max_{j\in [m]\backslash \S^\star}|\a_j^\top\bar{\r}_i|<\frac{c_1(\gamma)(1-\sqrt{\frac{k}{n}}-\delta)}{\sqrt{k}}
	\end{equation}
	is a sufficient condition for successful identification of an element from $\R_{psg}^{(i)}\cap (\S^\star\backslash \S^{(i)}_{psg})$.
	Note that since by the hypothesis of the inductive argument $\R_{psg}^{(i)} \cap (\S^\star\backslash\S_{psg}^{(i)}) \neq \emptyset$ and $\S_{psg}^{(i)} \subseteq \S^\star$ hold, $\bar{\r}_i$ is in the span of $\A_{\S^\star}$, and subsequently $\bar{\r}_i$ and $\a_j$'s are statistically independent for all $j\in [m]\backslash \S^\star$. Therefore, by Lemma \ref{lem:tropp} and the fact that $\a_j$'s are i.i.d. normal random vectors
	\begin{equation}\label{eq:counterpart}
		\begin{aligned}
			&\Pr\left(\max_{j\in [m]\backslash \S^\star}|\a_j^\top\bar{\r}_i| <\frac{c_1(\gamma)(1-\sqrt{\frac{k}{n}}-\delta)}{\sqrt{k}}\right)\\ & = \Pr\left(|\a_1^\top\bar{\r}_i| <\frac{c_1(\gamma)(1-\sqrt{\frac{k}{n}}-\delta)}{\sqrt{k}}\right)^{(m-k)} \\
			&\geq \left(1-\exp\left(-c_1(\gamma)^2(1-\sqrt{\frac{k}{n}}-\delta)^2\frac{n}{2k}\right)\right)^{(m-k)}\\
			&:=\tilde{q}_2^{\frac{1}{k}}
		\end{aligned}
	\end{equation}
	Finally, noting $\prod_{i=0}^{k-1} q_{psg}^{(i)} \geq \tilde{q}_1\prod_{i=0}^{k-1}\tilde{q}_2^{\frac{1}{k}} = \tilde{q}_1\tilde{q}_2$ establishes the stated results.
	
	We now proceed with the reminder of proof of Theorem 1. Let us take a closer look to $\tilde{q}_1$. We may bound $\tilde{q}_1$ using the inequality $(1-x)^l\geq 1-lx$, valid for $x\leq 1$ and $l\geq 1$ according to
	\begin{equation}
		\tilde{q}_1 \geq 1-2m\exp\left(-n (\frac{\gamma^2}{4}-\frac{\gamma^3}{6})\right)-\exp(-\delta^2\frac{n}{2}).
	\end{equation}
	Since our goal is to show the optimal sample complexity is achieved by \textsc{Psg}, comparing $\tilde{q}_1$ and $\tilde{q}_2$ we can conclude $\tilde{q}_1$ can be easily excluded from our numerical approximations as the exponent in $\tilde{q}_1$ increases linearly with $n$ while exponent in $\tilde{q}_2$ increases fairly more slowly. Alternatively, we can multiply $\tilde{q}_1$ and $\tilde{q}_2$, and by discarding positive higher order terms achieve the same conclusion.
	
	Now, lets turn our attention towards the lower bound on $\prod_{i = 0}^{k-1}p_{psg}^{(i)}$. From Theorem 3 of the main paper,
	\begin{equation}
		\prod_{i = 0}^{k-1}p_{psg}^{(i)}\geq (1-\epsilon)^k \geq 1-k\epsilon.
	\end{equation}
	Next, we find a simple lower bound on $\tilde{q}_2$. Assume, $(1-\sqrt{\frac{k}{n}}-\delta)^2 \geq 1-c$ for some $c>0$. Then, it holds that $n\geq C_2 k $, where $C_2 := (1-\sqrt{1-c}+\delta)^{-2}$. Thus, employing $(1-x)^l\geq 1-lx$ once again yields
	\begin{equation}
		\tilde{q}_2  \geq  1-k(m-k)\exp\left(-\frac{1-\gamma}{1+\gamma}(1-c)\frac{n}{2k}\right).  
	\end{equation}
	Let $C_1 := \frac{1-\gamma}{1+\gamma}\frac{1-c}{2}$. Given that $k(m-k)\leq\frac{1}{4}(\frac{m}{k})^6$ for $m>k\sqrt{k}$, we obtain
	\begin{equation}
		\tilde{q}_2  \geq  1-\frac{1}{4}(\frac{m}{k})^6\exp\left(-C_1\frac{n}{k}\right).  
	\end{equation}
	Now, since $(1-\beta)^2 \geq 1-2\beta$, in order to establish $\Pr\left(\S_{psg}^{(k)} = \S^\star\right) \geq 1-2\beta$, it suffices to show
	\begin{equation}
		1-k\epsilon >1-\beta, \text{ and}\qquad 1-\frac{1}{4}(\frac{m}{k})^6\exp\left(-C_1\frac{n}{k}\right) >1-\beta.
	\end{equation}
	Therefore, the condition on $\epsilon$, i.e., $\epsilon < \frac{\beta}{k}$, and the results emerge by rearranging the above inequalities.
\end{proof}
\end{document}

%% file: defspack.tex
\usepackage{amsmath,graphicx,multicol}
\usepackage{amsfonts}
\usepackage{color}
\usepackage{bbm}
\usepackage{cite}
\usepackage{amssymb}
\usepackage{algorithm}
\usepackage{algorithmic}
\usepackage{tabularx}
\usepackage{hyperref}
\usepackage[utf8]{inputenc}
\usepackage[english]{babel}

\usepackage{amsthm}
\usepackage{subcaption}
\usepackage[font={footnotesize}]{caption}
\usepackage{url}
\usepackage{mathtools}
\usepackage{comment}
\usepackage{array,booktabs,dcolumn,calc}
\usepackage{float}

\newtheorem{definition}{Definition}[]
\newtheorem{proposition}{Proposition}[]
\newtheorem{theorem}{Theorem}[]

\newtheorem{corollary}{Corollary}[theorem]
\newtheorem{lemma}[]{Lemma}

\newcommand{\ts}{\textsuperscript}

\DeclareMathOperator{\E}{\mathbb{E}}

\def\x{{\mathbf x}}
\def\a{{\mathbf a}}
\def\A{{\mathbf A}}
\def\E{{\mathbb E}}
\def\P{{\mathbf P}}
\def\r{{\mathbf r}}
\def\u{{\mathbf u}}
\def\y{{\mathbf y}}
\def\I{{\mathbf I}}
\def\e{{\boldsymbol{\nu}}}
\def\Rb{{\mathbb{R}}}
\def\C{{\mathbf C}}

\def\O{{\cal O}}
\def\S{{\cal S}}
\def\T{{\cal T}}

\def\Xs{{\cal X}}
\def\R{{\cal R}}
\def\N{{\cal N}}
\def\Z{{\cal Z}}

\def\B{{\cal B}}
\def\As{{\cal A}}
\def\Ps{{\cal P}}
\def\C{{\cal C}}

\let\oldref\ref
\renewcommand{\ref}[1]{(\oldref{#1})}
\newcommand{\RNum}[1]{\uppercase\expandafter{\romannumeral #1\relax}}
\newcolumntype{C}{>{\centering\arraybackslash}b{\widthof{positions}}}
\newcolumntype{d}{D{.}{.}{-2}}

\makeatletter
\renewcommand{\fnum@figure}{Fig.~\thefigure}
\makeatother

%% file: Secintro.tex
We study the problem of maximizing non-decreasing weak submodular functions under a cardinality constraint in large-scale settings.  The well-known \textsc{Greedy} algorithm \cite{nemhauser1978analysis} selects the solution set by sequentially identifying elements with the largest marginal contribution; the algorithm achieves a $1-1/e$ worst-case approximation \cite{nemhauser1978analysis}, the tightest guarantee for any algorithm that can evaluate the  objective function on only polynomially many inputs. Although \textsc{Greedy} achieves the optimal approximation factor, the computational cost of doing so is expensive for large-scale problems. This motivates the search for approximation schemes capable of accelerating the optimization without significant sacrifice of accuracy.

Recently, greedy algorithms that utilize random sampling while restricting size of the searched set to a fixed value have been proposed; for a problem with cardinality constraint $k$ and the ground set of size $m$, such methods incur complexity of only $\O(m\log\frac{1}{\epsilon})$  \cite{mirzasoleiman2015lazier,elenberg2016restricted,hashemi2017randomized,hashemi2020randomized}. In expectation, the fixed search size methods achieve a constant factor approximation of $1-1/e-\epsilon$, nearly matching the worst-case performance guarantee of \textsc{Greedy} while providing a computational gain of $\O(\frac{k}{\log\frac{1}{\epsilon}})$. Motivated by the success of greedy search schemes in practical settings, in this paper we investigate the impact of the search space size on their performance and study performance-complexity tradeoffs in maximizing a (weak) submodular function. To this end, we consider two criteria: (1) the ability of greedy algorithms with uniform sampling to exactly identify the optimal solution to a (weak) submodular maximization problem, and (2) the degree of suboptimality of the selected solution with respect to the optimal value.

Our first contribution, formalized in Theorem \ref{thm:fail}, is the somewhat surprising result that as the size of the ground set and cardinality constraint increase, randomized greedy schemes with a restricted search space with overwhelming probability fail to successfully identify the optimal subset. This in turn implies that, while there may be scenarios where \textsc{Greedy} can identify the optimal subset with high-probability (e.g., the task of sparse recovery), there is an unbounded gap between the exact identification capability of \textsc{Greedy} and randomized schemes with a restricted search space.

Aiming to overcome the above limitation, as part of our next contribution we establish that having an increasing schedule of sampling set sizes is unavoidable. Building on this insight, we propose a new algorithm that we refer to as {\it Progressive Stochastic Greedy} (\textsc{Psg}) and analyze its achievable performance. In particular, we show that \textsc{Psg} attains improved worst-case approximation factor, both on expectation and with high probability, compared to randomized greedy schemes with a fixed sampling size. 

Finally, we consider the application of the proposed scheme to two subset selection problems, namely, column subset selection for subspace clustering and observation selection for extended object tracking in automated driving. Our results demonstrate the efficacy of the proposed scheme in reducing the computational costs of subset selection with negligible performance drop.

\subsection{Related Work}
Submodularity is a property of set functions with desirable theoretical and practical implications relevant to many problems in combinatorial optimization. For instance, submodular maximization applies to many well-known problems such as facility location, coverage problems, and maximum weighted matching in discrete optimization \cite{williamson2011design} as well as active learning, influence maximization, and information gathering in machine learning \cite{kempe2003maximizing,krause2008near,guillory2012active}. In such problems, the goal is to maximize a monotonically increasing submodular function subject to a linear matroid, or a cardinality constraint. 

The objective function in some applications, e.g., sparse support selection and observation selection \cite{tropp2007signal,joshi2009sensor,das2011submodular,elenberg2016restricted}, is not necessarily a submodular function; rather, one deals with {\it weakly} submodular objectives that resemble diminishing return property of submodular functions. 

Recent advances in information systems have brought forth unprecedented amounts of data in many settings, including contemporary weak submodular maximization problems. Given a cardinality constraint $k$ and a ground set of size $m$, the classical \textsc{Greedy} algorithm for monotone weak submodular maximization that enjoys an optimal $1-1/e$ constant factor approximation \cite{nemhauser1978analysis} requires $\O(mk)$ function evaluations. Therefore, in data intensive applications where function evaluation is expensive, running \textsc{Greedy} may be infeasible. To this end, there have been recent efforts to exploit strong theoretical guarantees of \textsc{Greedy} while improving on its complexity via resorting to either distributed and parallel computing schemes, or methods for reducing the cost-per-iteration of \textsc{Greedy}. Among the former, there is a growing line of work to design algorithms with sublinear {\it adaptivity} \cite{mirzasoleiman2013distributed,balkanski2018adaptive,ene2019submodular,fahrbach2018submodular}. The concept of adaptivity is heavily studied in computer science and optimization; informally, adaptivity characterizes efficiency of parallel computation of an algorithm. The focus of this paper, however, is on the latter, i.e., centralized schemes -- distributed weak submodular maximization methods are complementary to our study. Nevertheless, our analysis and the proposed algorithm may be deployed to aid distributed methods that rely on \textsc{Greedy}, potentially extending their utility.

The \textsc{Lazy-Greedy} algorithm \cite{minoux1978accelerated}  exploits the notion of  submodularity to decrease the number of function evaluations of each iteration of \textsc{Greedy} without sacrificing its performance. However, similar to \textsc{Greedy}, \textsc{Lazy-Greedy} incurs $\O(mk)$ function evaluations. Moreover, it cannot be employed in weak submodular maximization problems. More recently, Badanidiyuru and Vondrak \cite{badanidiyuru2014fast} proposed a randomized scheme that achieves a worst case approximation factor of $1-1/e-\epsilon$ while using $\O(\frac{m}{\epsilon}\log\frac{m}{\epsilon})$ evaluations. Motivated by this work, Mirzasoleyman et al. \cite{mirzasoleiman2015lazier} proposed \textsc{Stochastic-Greedy} that achieves a worst case approximation factor of $1-1/e-\epsilon$ while using  $\O(m\log\frac{1}{\epsilon})$ function evaluations. Further discussion on greedy algorithms with random sampling is deferred to Section \ref{sec:grs}. To our knowledge, no prior works study the effect of the sampling size on the success probability of the greedy algorithms with random sampling.

Sparse reconstruction and sparse support selection tasks belong to a class of cardinality-constrained weak submodular maximization problems where the exact identification of the optimal subset is of critical importance. In sparse support selection, the goal is to identify the support of a high dimensional vector (e.g., an image or a signal), i.e., the collection of nonzero components of the data, from a relatively small number of  measurements. In such settings, \textsc{Greedy} satisfies a general constant factor approximation guarantee as shown by \cite{nemhauser1978analysis,das2011submodular}. However, by exploiting the underlying structural properties of the measurement model in sparse support selection, one can establish conditions under which \textsc{Greedy} {\it exactly} identifies the optimal subset. To this end, necessary and sufficient conditions for exact identification via
\textsc{Greedy} have been established by relying on various analysis techniques including those based on restricted isometry 
\cite{zhang2011sparse,davenport2010analysis,mo2012remark} and mutual incoherence
\cite{tropp2004greed,cai2011orthogonal,zhang2009consistency} properties. In particular, when the measurements are randomly generated, Tropp and Gilbert \cite{tropp2007signal} show that  \textsc{Greedy} enjoys an optimal sample complexity bound outlined by \cite{candes2005decoding}. In the supplementary, we explore an application of the proposed algorithm to this task.

A related task of column subset selection (CSS) has received considerable attention in recent years due to its broad applicability, interpretability, and provably-guaranteed performance \cite{tropp2009column,guruswami2012optimal,boutsidis2014near}. CSS is a constrained low-rank-approximation problem that seeks to approximate a data matrix (e.g., a matrix having data points in its rows and features in its columns) by projecting it onto a space spanned by only a few of its columns. While similar to the general low-rank approximation problem, CSS possesses certain distinguishing characteristics. First, since CSS is an unsupervised
method and does not require labeled data, it can be applied efficiently to the scenarios where labeled data is sparse while unlabeled data is abundant. Second, to learn interpretable models in applications where the decision is made via a data-driven algorithm (e.g., hiring and education), it is of critical importance to keep the semantic interpretation of the features intact. This can be ensured by selecting a subset of available features as opposed to generating new features via an arbitrary function of the input features. Finally, compared to PCA or other methods that require a matrix-matrix multiplication to project input features into a reduced space during inference time, solution of CSS feature selection problem can be applied efficiently during inference as CSS only requires selecting a subset of feature values from a new instance vector.

An efficient approach to CSS relies on a greedy algorithm that selects the most representative subset of columns in an iterative fashion by greedily optimizing the reconstruction error \cite{farahat2013distributed,farahat2015greedy}. The iterative procedure of the greedy scheme is readily implemented in practice, and its often strong performance is complemented with rigorous theoretical guarantees \cite{civril2012column,altschuler2016greedy}. However, running the greedy scheme can be computationally expensive for large datasets. This is because if the goal is to select $k$ out of the $m$ available columns, in each of $k$ iterations of the greedy scheme one needs to evaluate marginal contribution of $\O(n)$ columns. Although computational costs can be reduced using the so-called lazy evaluations \cite{krause2014submodular}, the worst case number of function evaluations of  the greedy scheme is $\O(nk)$. In our work, we apply the proposed \textsc{Psg} algorithm to improve the computational efficiency of greedy CSS schemes.
\subsection{Organization}
The rest of the paper is organized as follows. In Section \ref{sec:back}, we introduce the notation and review relevant concepts from weak submodular maximization.  In Section \ref{sec:fail}, we study the conditions for successful identifications of the optimal subset via uniform sampling strategies. In particular, we present  our first contribution which establishes the failure of methods with a fixed sampling size to exactly identify the optimal subset, and argue the necessity of having an increasing schedule of sampling sizes; this insight leads to the proposed \textsc{Psg} algorithm. Section \ref{sec:alg}  presents the analysis of \textsc{Psg} and provides bounds on its suboptimality performance. In Section \ref{sec:sim}, we consider applications of \textsc{Psg} in subset selection problems. Finally, concluding remarks are provided in Section \ref{sec:conc}.

%% file: Secback.tex
In this section, we introduce the notation and provide an overview of relevant concepts from submodular optimization as well as the \textsc{Greedy} algorithm and methods with uniform sampling.
\subsection{Notation}
Italic letters represent scalars and numerical constants, e.g., $\alpha$ and $C$. We use calligraphic letters to denote sets, e.g., $\S$. Bold capital letters denote matrices, e.g., $\A$, while bold lowercase letters represent column vectors, e.g., $\a$.  
Finally, $\mathbb{I}(.)$ denotes the indicator function of its argument.
\subsection{Submodular Maximization}
\begin{definition}
\label{def:mon}
Set function $f:2^\Xs\rightarrow \mathbb{R}$  is monotone if $f(\S)\leq f(\T)$ for all $\S\subseteq \T\subseteq \Xs$.
\end{definition}

\begin{definition}
\label{def:submod}
Set function $f:2^\Xs\rightarrow \mathbb{R}$ is submodular if 
\begin{equation}
f(\S\cup \{j\})-f(\S) \geq f(\T\cup \{j\})-f(\T)
\end{equation}
for all subsets $\S\subseteq \T\subset \Xs$ and $j\in \Xs\backslash \T$. The term $f_j(\S)=f(\S\cup \{j\})-f(\S)$ is the marginal value of adding element $j$ to set $\S$.
\end{definition}

Given a monotone non-decreasing set function $f:2^\Xs\rightarrow \mathbb{R}$ with $f(\emptyset)=0$, we are interested in solving the combinatorial optimization problem
\begin{equation}\label{eq:fmax}
\begin{aligned}
& \underset{\S}{\text{maximize}}
\quad f(\S)\\
& \text{subject to}\hspace{0.4cm}  \S \subset \Xs, \phantom{k}|\S|\leq k,
\end{aligned}
\end{equation}
which we denote by $\mathcal{P}(m,k)$, where $|\Xs| = m$.
By a reduction to the well-known set cover problem, the combinatorial optimization \eqref{eq:fmax}
can be shown to be NP-hard \cite{feige1998threshold,williamson2011design}. 
It has been shown  in \cite{nemhauser1978analysis} that if $f(\cdot)$ is monotone and submodular, a simple greedy algorithm that iteratively selects an element with the highest marginal gain (see Algorithm \ref{alg:g}) satisfies the optimal $1-1\slash e$ worst case approximation ratio.

\subsection{Weak Submodularity}
In many problems, the objective function is not submodular but under certain conditions it behaves similarly. Such functions are called {\it weakly} submodular and the extent of their proximity to submodularity is captured using the following parameters.

\begin{definition}
The multiplicative weak-submodularity constant of a monotone non-decreasing function $f$ is defined as
 \begin{equation}
 c_{f}={\max_{(\S,\T,i)\in \tilde{\Xs}}{f_i(\T)\slash f_i(\S)}},
 \end{equation}
where $\tilde{\Xs} = \{(\S,\T,i)|\S \subseteq \T \subset \Xs, i\in \Xs \backslash \T\}$. 
\end{definition}
The multiplicative weak-submodularity constant \cite{zhang2016submodular,chamon2017approximate,hashemi2019submodular} is a closely related concept to submodularity and essentially quantifies how close the set function is to being submodular. It is worth noting that a set function $f(\S)$ is submodular if and only if its multiplicative weak-submodularity constant satisfies 
$c_{f} \le 1$ \cite{das2011submodular,elenberg2016restricted,horel2016maximization}.

A similar notion of weak submodularity is the additive weak-submodularity constant defined below \cite{zhang2016submodular,chamon2017approximate,hashemi2019submodular}.

\begin{definition}
The additive weak-submodularity constant of a monotone non-decreasing function $f$ is defined as
 \begin{equation}
 \epsilon_{f}={\max_{(\S,\T,i)\in \tilde{\Xs}}{f_i(\T)- f_i(\S)}},
 \end{equation}
where $\tilde{\Xs} = \{(\S,\T,i)|\S \subseteq \T \subset \Xs, i\in \Xs \backslash \T\}$. 
\end{definition}
Note that when $f(\S)$ is submodular, its additive weak-submodularity constant satisfies $\epsilon_{f}\leq 0$.

For a monotone function with bounded additive and multiplicative weak-submodularity constants (WSCs) we have the following proposition (see, e.g., \cite{hashemi2019submodular}). 
\begin{proposition}\label{lem:curv}
Let $c_f$ and $e_f$ be the multiplicative and  additive weak-submodularity constants of $f(\S)$, a monotone non-decreasing function with $f(\emptyset)=0$, respectively. Let $\S$ and $\T$ be any subsets such that $\S\subset \T \subseteq \Xs$ with $|\T\backslash \S|=r$. Then, it holds that
\begin{equation}
f(\T)-f(\S)\leq \frac{1}{r}\left(1+(r-1)c_f\right)\sum_{j\in \T\backslash \S}f_j(\S),
\end{equation}
and 
\begin{equation}
f(\T)-f(\S)\leq (r-1)e_f+\sum_{j\in \T\backslash \S}f_j(\S).
\end{equation}
\end{proposition}

It is worth pointing out other weak submodularity notions such as those in \cite{das2011submodular,elenberg2016restricted,zhang2016submodular,horel2016maximization} that depending on the application may simplify the derivation of the approximation bounds (see e.g., \cite{bian2017guarantees,chamon2017approximate,khanna2017scalable,hashemi2019submodular,ghasemi2019submodularity}).

\begin{algorithm}[t]
	\caption{\textsc{Greedy}}
	\label{alg:g}
	\begin{algorithmic}[1]
		\STATE \textbf{Input:} Weak submodular function $f$, ground set $\Xs$, number of  elements to be selected $k$.
		\STATE \textbf{Output:} Subset $\S_g\subseteq \Xs$ with $|\S_g|=k$.
		\STATE Initialize $\S_g^{(0)} = \emptyset$
		\FOR{$i = 0,\dots, k-1$}
		\STATE $j_s \in \text{argmax}_{j \in \Xs} f_j(\S_g^{(i)})$
		\STATE $\S_g^{(i+1)} = \S_g^{(i)}  \cup \{j_s\}$
		\ENDFOR
		\RETURN $\S_g = \S_g^{(k)}$.
	\end{algorithmic}
\end{algorithm}

Using notion of weak-submodularity, one can extend the theoretical results of \cite{nemhauser1978analysis} for \textsc{Greedy} to the case of weak submodular functions \cite{das2011submodular}. In particular, the following proposition holds (see, e.g. \cite{hashemi2019submodular}).
 
\begin{theorem}\label{thm:curv}
Let $c_f$ and $e_f$ be the multiplicative and additive weak-submodularity constants of $f(\S)$, a monotone non-decreasing function with $f(\emptyset)=0$. Let $\S_g \subseteq \Xs$ with $|\S_g|\leq k$ be the subset selected when maximizing $f(\S)$ subject to a cardinality constraint  via the greedy observation selection scheme, and let $\S^\star$ denote the optimal subset. Then
\begin{equation}
f(\S_g) \geq \left(1-e^{-\frac{1}{c}}\right) f(\S^\star),
\end{equation}
where $c=\max\{c_f,1\}$\footnote{Henceforth, we assume $c_f\geq 1$ to emphasize that the objective function is typically weak submodular.} and
\begin{equation}
f(\S_g) \geq \left(1-\frac{1}{e}\right) \left(f(\S^\star)-(k-1)e_f\right).
\end{equation}
\end{theorem}

The approximation results in Proposition \ref{thm:curv} imply that if the objective function is monotone and weak submodular, the greedy selection scheme which in each iteration selects an element with the highest marginal gain finds a solution that is close to the optimal.

\subsection{\textsc{Greedy} with Random Sampling}\label{sec:grs}

If $|\Xs| = m$, in each of $k$ iterations of \textsc{Greedy} one needs to find the marginal gain of $\O(m)$ elements. This is a computationally intensive procedure when the involved datasets are large. Although the computational costs can be reduced using the so-called lazy evaluations \cite{krause2014submodular}, the worst case number of function evaluations of \textsc{Greedy} is $\O(mk)$. The prohibitive complexity of \textsc{Greedy} for large-scale datasets has motivated the design of more efficient schemes for weak submodular maximization. A body of work typically referred to as \textsc{Stochastic-Greedy} algorithms, aims to reduce the number of function evaluations by restricting the search domain in each iteration of the greedy selection procedure 
\cite{mirzasoleiman2015lazier,elenberg2016restricted,hashemi2017randomized,hashemi2020randomized} (see Algorithm \ref{alg:sgp} for a general template). Specifically, instead of evaluating the marginal gain of $\O(m)$ elements, in the $i\ts{th}$ iteration of greedy-with-random-sampling one selects a subset $\R^{(i)} \subseteq \Xs$ by uniformly at random sampling $r_i$  elements and only evaluates marginal gains of the elements in $\R^{(i)}$. References \cite{mirzasoleiman2015lazier,elenberg2016restricted,hashemi2017randomized,hashemi2020randomized} explore the case where $\R^{(i)} = \R$ is fixed and $r_i = r \frac{m}{k}\log\frac{1}{\epsilon}$, where the parameter $\epsilon$, $0<e^{-k}\leq \epsilon \leq e^{-\frac{k}{m}}<1$, determines the size of the search domain $\R^{(i)}$ and thus controls the number of function evaluations in each iteration.\footnote{In this paper, for simplicity of presentation we assume $\log\frac{1}{\epsilon}$ and $r$ are integer-valued quantities.}  It turns out that with this choice, the complexity of greedy-with-random-sampling is $\O(m\log\frac{1}{\epsilon})$ and that it selects a subset $\S_{sg}$ such that
\begin{equation}\label{sgrbound}
\E[f(\S_{sg})] \geq \left(1-1/e-\epsilon\right) f(\S^\star),
\end{equation}
 given that $f$ in \eqref{eq:fmax} is submodular \cite{mirzasoleiman2015lazier}. This approximation ratio is derived under the simplifying assumption that the sequence of random subsets $\{\R_{sg}^{(i)}\}_{i=0}^{k-1}$ is constructed via sampling with replacement. Khanna et al. \cite{khanna2017scalable} analyze \textsc{Stochastic-Greedy} for weak submodular functions and show that  it achieves an {\it expected} $1-e^{-1/c}-\epsilon$ worst case approximation ratio. For the setting where $\{\R_{sg}^{(i)}\}_{i=0}^{k-1}$ are constructed via sampling without replacement, this approximation ratio is improved to $1-e^{-1/c}-\epsilon^\beta/c$, where $\beta = 1+\O(1/k)$ \cite{hashemi2017randomized}. Under a specific set of assumptions on the marginal gain of elements selected in each iteration, \cite{hashemi2020randomized} have shown that a similar result holds not only on expectation, but with high probability.
\begin{algorithm}[t]
	\caption{\textsc{Greedy} with Random Sampling}
	\label{alg:sgp}
\begin{algorithmic}[1]
		\STATE \textbf{Input:} Weak submodular function $f$, ground set $\Xs$, number of  elements to be selected $k$, search space schedule $\{r_i\}_{i=0}^{k-1}$.
		\STATE \textbf{Output:} Subset $\S^{(k)}\subseteq \Xs$ with $|\S^{(k)}|=k$.
		\STATE Initialize $\S^{(0)} = \emptyset$
		\FOR{$i = 0,\dots, k-1$}
		\STATE Form $\R^{(i)}$ by sampling $\min (r_i,m)$ elements from $\Xs$ uniformly at random.
		\STATE $j_s \in \text{argmax}_{j \in \R^{(i)}} f_j(\S^{(i)})$
		\STATE $\S^{(i+1)} = \S^{(i)}  \cup \{j_s\}$
		\ENDFOR
	\end{algorithmic}
\end{algorithm}

%% file: Secpre.tex
We start the analysis by studying the probability of successfully identifying the optimal subset $\S^\star$. This notion is formalized in the following definition.
\begin{definition}
Let $\textsc{Alg}$ be an approximation algorithm for the weak submodular optimization problem \eqref{eq:fmax} with a unique solution $\S^\star$.  Let $\S_{alg}$ be the output of $\textsc{Alg}$. Then, $\textsc{Alg}$ successfully identifies $\S^\star$ if $\S_{alg} = \S^\star$. Furthermore, the probability of success of $\textsc{Alg}$ is defined as $\Pr\left(\S_{alg} = \S^\star\right)$.
\end{definition}
Our goal in this section is to quantify the impact of $r_i$ on the performance of Algorithm \ref{alg:sgp} by analyzing its probability of success. To successfully identify $\S^\star$, in each iteration of Algorithm \ref{alg:sgp} at least one new (not previously selected) element of $\S^\star$ should be present in the randomly selected subset $\R^{(i)}$. More formally, if $\S^{(i)}$ denotes the subset of elements selected by Algorithm \ref{alg:sgp} before executing the $i\ts{th}$ iteration, $i = 0,\dots,k-1$, the set $\R^{(i)} \cap (\S^\star\backslash\S^{(i)})$ should be nonempty. This, however, is not sufficient -- to find the optimal solution, Algorithm \ref{alg:sgp} must in each iteration select elements from $\R^{(i)} \cap (\S^\star\backslash\S^{(i)})$. Since $|\S^\star| = k$ and since in each iteration Algorithm \ref{alg:sgp} selects one element, if there exists an $i\in [k]$ such that $\S^{(i)} \not\subseteq \S^\star$, then Algorithm \ref{alg:sgp} fails to identify $\S^\star$. 
This informal argument underlies Lemma \ref{lem:succ} below which will allow us to characterize the success probability.

\begin{lemma}\label{lem:succ}
Suppose the optimal solution of \eqref{eq:fmax} is unique. Let $\S^{(k)}$ denote the subset selected by Algorithm \ref{alg:sgp}, and let $\R^{(i)}$ denote the randomly selected search space of Algorithm \ref{alg:sgp} in the $i\ts{th}$ iteration. Then, it holds that
\begin{equation}\label{eq:succrel}
\begin{aligned}
    \Pr\left(\S^{(k)} = \S^\star\right)
    &=\prod_{i = 0}^{k-1} p^{(i)}\prod_{i = 0}^{k-1}q^{(i)},
    \end{aligned}
\end{equation}
where
\begin{equation}\label{eq:pi}
    p^{(i)} = \Pr\left(\R^{(i)} \cap (\S^\star\backslash\S^{(i)}) \neq \emptyset \text{ }| \text{ }\S^{(i)} \subset  \S^\star\text{ },\text{ } |\S^{(i)}| = i\right),
\end{equation}
and
\begin{equation}\label{eq:qi}
    q^{(i)} =\Pr\left(\S^{(i+1)} \subset \S^\star \text{ }|\text{ }\R^{(i)} \cap (\S^\star\backslash\S^{(i)}) \neq \emptyset , \text{ }|\S^{(i)}| = i\right).
\end{equation}
\end{lemma}
\begin{proof}
Let $\As^{(i)}$ denote the event $\{\S^{(i+1)} \cap \S^{(i)} \neq \emptyset, \S^{(i+1)} \subseteq  \S^\star\}$. Then the probability of success of  Algorithm \ref{alg:sgp} can be expressed as
\begin{equation}
\begin{aligned}
    \Pr\left(\S^{(k)} =   \S^\star\right) &= \Pr\left(\cap_{i = 0}^{k-1} \As^{(i)}\right)\\
    &= \prod_{i = 0}^{k-1}\Pr\left(\As^{(i)}\text{ }|\text{ }\cap_{j = 0}^{i-1}\As^{(j)}\right)\\
    &= \prod_{i = 0}^{k-1}\Pr\left(\As^{(i)}\text{ }|\text{ }\B^{(i)}\right),
\end{aligned}
\end{equation}
where $\B^{(i)}=\{\S^{(i)} \subset  \S^\star, |\S^{(i)}| = i\}$. Note that $\As^{(i)}$  can equivalently be written as 
\begin{equation}
    \As^{(i)} = \{\R^{(i)} \cap (\S^\star\backslash\S^{(i)}) \neq \emptyset, \S^{(i+1)} \subseteq  \S^\star\}.
\end{equation}
By further conditioning, this can be re-written as
\begin{equation}
\begin{aligned}
    \Pr\left(\S^{(k)} = \S^\star\right)
    &=\prod_{i = 0}^{k-1} p^{(i)}\prod_{i = 0}^{k-1}q^{(i)},
    \end{aligned}
\end{equation}
where $p^{(i)}$ and $q^{(i)}$ are given by \eqref{eq:pi} and \eqref{eq:qi}, respectively.

\end{proof}
Lemma \ref{lem:succ} demonstrates that the probability of success of  Algorithm \ref{alg:sgp} is product of two terms: (i) p = $\prod_{i = 0}^{k-1} p^{(i)}$ that characterizes the likelihood of the event $\mathcal{E}_1$ that $\R^{(i)}$ contains at least one new element of $\S^\star$ for all $i = 0,\dots,k-1$, (ii) $q = \prod_{i = 0}^{k-1} q^{(i)}$ that characterizes the likelihood of the event $\mathcal{E}_2$  of selecting one of the elements in the nonempty intersection of search space $\R^{(i)}$ and $\S^\star\backslash \S^{(i)}$. The events $\mathcal{E}_1$ and $\mathcal{E}_2$ collectively are necessary and sufficient conditions for exact identification of the optimal subset.

The first probability, $p =\prod_{i = 0}^{k-1} p^{(i)}$, is of particular interest as it can be thought of as being a general upper bound on the probability of success. Following this idea, next we establish an upper bound on the asymptotic probability of success and show that, somewhat surprisingly, for large-scale problems some variants with overwhelming probability fail to successfully identify the optimal subset.
 

\begin{theorem}\label{thm:fail}
Consider a sequence of optimization problems $\mathcal{P}(m,k)$ in \eqref{eq:fmax} with increasingly higher dimensions, i.e., the setting where ${m,k\rightarrow \infty}$, $m>k$. 
Let \textsc{Alg} denote a variant of Algorithm \ref{alg:sgp}
with a restricted uniform search space $\R^{(i)} \subset [m]$ having 
cardinality $r_i$ such that $\max_i r_i < m-k$ and $\limsup_{m,k\rightarrow \infty} r_i/m = 0$, for all $i$. Then the probability that \textsc{Alg} succeeds on $\mathcal{P}(m,k)$ goes to zero, i.e.,\footnote{Note that $\S_{alg}^{(k)}$ and $\S^\star$ are quantities that depend on $m$.}
			\begin{equation}
		\limsup_{m,k\rightarrow \infty} 	\Pr\left(\S_{alg}^{(k)} = \S^\star\right) =0.
		\end{equation} 
\end{theorem} 
\begin{proof}
First, assume \textsc{Alg} uses sampling with replacement to construct $\R_{alg}^{(i)}$. We can compute $p_{alg}^{(i)}$ according to
	\begin{equation}
	\begin{aligned}
	p_{alg}^{(i)} &= 1-\Pr\left(\R_{alg}^{(i)} \cap (\S^\star\backslash\S_{alg}^{(i)}) = \emptyset | \B_{alg}^{(i)}\right) \\
	&= 1-\left(1-\frac{|\S^\star\backslash \S_{alg}^{(i)}|}{|\Xs|}\right)^{r}\\
	&= 1-\left(1-\frac{k-i}{m}\right)^{r}.
	\end{aligned}
	\end{equation}
	Note that since $p_{alg}^{(i)}\leq 1$, it follows that
	\begin{equation}\label{eq:bound1}
	\prod_{i = 0}^{k-1} p_{alg}^{(i)} \leq p_{alg}^{(k-1)} = 1-\left(1-\frac{1}{m}\right)^{r}.
	\end{equation}
Therefore, since $\max_i r_i < m-k$ and $\limsup_{m,k\rightarrow \infty} r_i/m = 0$, for all $i$, we can establish
\begin{equation}
\begin{aligned}
    \limsup_{m,k\rightarrow \infty} 	\Pr\left(\S_{alg}^{(k)} = \S^\star\right) &\leq \limsup_{m,k\rightarrow \infty}  \prod_{i = 0}^{k-1} p_{alg}^{(i)} \leq \limsup_{m,k\rightarrow \infty}  p_{alg}^{(k-1)} \\&=  \limsup_{m,k\rightarrow \infty} 1-\left(1-\frac{1}{m}\right)^{r} = 0.
    \end{aligned}
\end{equation}
We next consider the case where $\R_{alg}^{(i)}$ is constructed by sampling the elements in $\Xs \backslash \S_{alg}^{(i)}$ without replacement. The probability of success in this case is higher than when sampling with replacement. Nevertheless, we can derive $p_{alg}^{(i)}$ according to
\begin{equation}
\begin{aligned}
p_{alg}^{(i)} &= 1-\prod_{l = 0}^{r_i-1}\left(1-\frac{k-i}{m-l}\right).
\end{aligned}
\end{equation}
To establish the asymptotic probability, we upperbound the success probability with $p_{alg}^{(k-1)}$. Once again, since $\limsup_{m,k\rightarrow \infty} r_i/m = 0$, for all $i$, one can observe that $\limsup_{m,k\rightarrow \infty}  p_{alg}^{(k-1)} $, thereby proving the stated result.
\end{proof}
Theorem \ref{thm:fail} establishes an upper bound on the probability that a variant of \textsc{Greedy} with a restricted search space constructed uniformly at random identifies $\S^\star$ exactly. The theorem states that as long as $\limsup_{m,k\rightarrow \infty}r_i/m = 0$, the asymptotic success probability is zero. To illustrate the implications of this theorem, we consider the scenario where $r_i = r$ is kept fixed, i.e., the size of the restricted search space is  constant. This choice is explored in \cite{mirzasoleiman2015lazier,elenberg2016restricted,hashemi2017randomized,hashemi2020randomized} where bounds on the expected approximation factor of Algorithm \ref{alg:sgp} are derived with $r = \frac{m}{k}\log \frac{1}{\epsilon}$. Corollary \ref{thm:fail2} below derives bound on the success probability under this consideration.

\begin{corollary}\label{thm:fail2}
Under the assumptions of Theorem \ref{thm:fail} and $r = \frac{m}{k}\log\frac{1}{\epsilon}$, consider a sequence of optimization problems $\mathcal{P}(m,k)$ as (\ref{eq:fmax}) where ${m,k\rightarrow \infty}$. 
The following claims hold:
	\begin{enumerate}
		\item If there exists $\alpha \in (0,1)$ such that $r \leq k^{\alpha-1} m$ (equivalently, $\epsilon \geq \exp(-k^\alpha)$), then the probability that \textsc{Alg} 
		succeeds on $\mathcal{P}(m,k)$ goes to zero, i.e.,
			\begin{equation}
		\limsup_{m,k\rightarrow \infty} 	\Pr\left(\S_{alg}^{(k)} = \S^\star\right) =0.
		\end{equation} 
		\item If there exists $\alpha_1 \in (0,1)$ such that $r \leq \alpha_1 m$ (equivalently, $\epsilon \geq \exp(-\alpha_1 k)$), then the probability that \textsc{Alg} 
		succeeds on $\mathcal{P}(m,k)$ satisfies
		\begin{equation}
		\limsup_{m,k\rightarrow \infty} 	\Pr\left(\S_{alg}^{(k)} = \S^\star\right) \leq 1-\exp\left(-\alpha_1\right)<0.64.
		\end{equation}
	\end{enumerate}
\end{corollary} 
\begin{proof}
For simplicity, we only consider the case of constructing $\R$ via sampling with replacement as the case of sampling without replacement can be treated analogously.

First, consider the setting where $\epsilon \geq \exp(-k^\alpha)$, $0<\alpha <1$.
In light of Theorem \ref{thm:fail},
to establish the first result it suffices to show that
\begin{equation}
    \limsup_{m,k\rightarrow \infty}  1-\left(1-\frac{1}{m}\right)^{r} = 0.
\end{equation}
Using Lemma \ref{lem:ineq} yields
\begin{equation}
\begin{aligned}
    1-\left(1-\frac{1}{m}\right)^{r} &\leq 1-\exp\left(-\frac{r}{m}\right)\left(1-\frac{r}{m^2}\right)\\
    &= 1-\exp\left(\frac{\log{\epsilon}}{k}\right) \left(1+\frac{\log{\epsilon}}{mk}\right)\\
    &\leq 1-\exp\left(-k^{\alpha-1}\right)\left(1-\frac{k^{\alpha-1}}{m}\right),
    \end{aligned}
\end{equation}
where for the last inequality we recall the assumption $\epsilon \geq \exp\left(-k^\alpha\right)$. The result is then established by noting $\limsup_{m,k\rightarrow \infty} \frac{k^{\alpha-1}}{m} = 0$, $\limsup_{k\rightarrow \infty} \exp\left(-k^{\alpha-1}\right) = 1$, and using the squeeze theorem.

Next, consider the second setting, i.e.,  $\epsilon \geq \exp(-\alpha_1 k)$, $0<\alpha_1 <1$. Following a similar approach, one obtains
\begin{equation}\label{eq:rk1}
\begin{aligned}
    \Pr\left(\S_{alg}^{(k)} = \S^\star\right) &\leq 1-\exp\left(-\alpha_1\right)\left(1-\frac{\alpha_1}{m}\right):=\delta_2.
    \end{aligned}
\end{equation}
Since the bound in \eqref{eq:rk1} goes to $1-\exp\left(-\alpha_1\right)$ as  $m,k \rightarrow \infty$, it holds that $\delta_2\leq 1-\exp\left(-\alpha_1\right) < 1-1/e < 0.64$.
\end{proof}
Corollary \ref{thm:fail2} establishes upper bounds on the probability that a variant of \textsc{Greedy} with a restricted search space constructed uniformly at random identifies $\S^\star$ exactly in two scenarios: (i) If the size of the 
search space remains fixed in each iteration of \textsc{Alg} and the algorithm makes $\O(mk^\alpha)$ oracle calls 
for some $\alpha \in (0,1)$, then the probability of the exact identification approaches zero as the problem dimension 
grows. (ii) If the size of the search space remains fixed in each iteration of \textsc{Alg} and strictly less than $[m]$, 
and the algorithm makes $\O(mk)$ oracle calls, then although the probability of the exact identification may be nonzero, it is not asymptotically $1$. Note that in many applications, including sparse reconstruction and sparse learning \cite{tropp2004greed,candes2005decoding}, an arbitrarily high success probability is a condition required to establish any nontrivial sample complexity results, i.e., the minimum number of data points for successful recovery and prediction. 
\subsection{Progressively-Increasing Random Sampling}
The proofs of Theorem \ref{thm:fail} and Corollary \ref{thm:fail2} reveal the underlying cause for the failure of Algorithm \ref{alg:sgp} to find the optimal solution: since the size of the search domain is fixed throughout the iterations, if Algorithm \ref{alg:sgp} successfully identifies elements from $\S^\star$ in earlier iterations, the chance of sampling new elements from $\S^\star$ significantly decreases in the subsequent ones. Therefore, success in earlier iterations increases the chance of failure to select new elements from $\S^\star$ in the subsequent ones. This phenomenon is not encountered in the \textsc{Greedy} algorithm since in each iteration \textsc{Greedy} considers all the elements in the ground set, including those in $\S^\star$. Therefore, although in initial iterations Algorithm \ref{alg:sgp} may search smaller domains, if the goal is to identify exactly all the elements in $\S^\star$, one should progressively increase the size of the search domain to improve the probability of success. 


 
We thus propose a simple \textit{increasing schedule} strategy of search spaces which grow to \textit{ultimately reach size $m$}. In particular, we progressively increase the size of the search domain as the cardinality of the identified subset $\S^{(i)}$ grows. The proposed method, referred to as {\it Progressive Stochastic Greedy} (\textsc{Psg}), in the $i^{th}$ iteration samples  $r_i=\frac{m}{k-i}\log\frac{1}{\epsilon}$ elements uniformly at random from $\Xs$ to construct the search set $\R^{(i)}$. Following \cite{mirzasoleiman2015lazier,elenberg2016restricted,hashemi2017randomized,hashemi2020randomized}, we let $\epsilon$, such that $e^{-k}\leq \epsilon \leq e^{-\frac{k}{m}}$, be a parameter which allows one to strike a desired balance between performance and complexity; the sampling may be with or without replacement. Note that since $r_i \leq m$ for all $i = 0,\dots,k-1$, for any iteration $i$ such that $i\geq i_s:= k -\log\frac{1}{\epsilon}$, we set $r_i$ to its maximum value, $m$. 
Thus, this procedure can be interpreted as a hybrid scheme that provides a soft transition from a restricted search space with random sampling to \textsc{Greedy}.    
\subsection{Complexity Analysis of \textsc{Psg}}
Recall that \textsc{Greedy} and Algorithm \ref{alg:sgp} with $r_i = r = \frac{m}{k}\log \frac{1}{\epsilon}$, i.e., \textsc{Stochastic-Greedy} \cite{mirzasoleiman2015lazier,elenberg2016restricted,hashemi2017randomized,hashemi2020randomized}, require $\O(mk)$ and $\O(m\log\frac{1}{\epsilon})$ function evaluations, respectively. In the following, we express the complexity of \textsc{Psg} in terms of the total number of function evaluations throughout the iterations of the algorithm. Specifically, given our choice of $r_i$, we have that
\begin{equation}\label{eq:cmp}
\begin{aligned}
\O&\left(\sum_{i=0}^{k-\log\frac{1}{\epsilon}}\frac{m}{k-i}\log\frac{1}{\epsilon}+m\log\frac{1}{\epsilon} \right) 
= \O\left(m\log(\frac{1}{\epsilon})\log k\right).
\end{aligned}
\end{equation}
As an example, for  $\epsilon \geq \exp(-k^\alpha)$ and $\epsilon \geq \exp(-\alpha k)$, $0<\alpha<1$, the complexity expression in \eqref{eq:cmp} reduces to $\tilde{\O}(mk^\alpha)$ and $\O(mk)$, respectively. Thus, the proposed method incurs at most a factor $\O(\log k)$ higher complexity than \textsc{Stochastic-Greedy}. As we show in the reminder of the next subsection, this relatively small increase in complexity is sufficient to satisfy a necessary condition for identifying $\S^\star$.
\subsection{Lower Bound on the Probability of Success}
We now aim to determine whether \textsc{Psg}, i.e., Algorithm \ref{alg:sgp} with the proposed progressively-increasing schedule, can identify the optimal solution to $\mathcal{P}(m,k)$ as $m,k \rightarrow \infty$. 

To answer the above question, one needs to establish a sufficient condition for the exact identification of $\S^\star$ or, equivalently, a lower bound on the probability of success of Algorithm \ref{alg:sgp}. Recall from Lemma \ref{lem:succ} that
\begin{equation}
\begin{aligned}
    \Pr\left(\S^{(k)} = \S^\star\right)
    &=\prod_{i = 0}^{k-1} p^{(i)}\prod_{i = 0}^{k-1}q^{(i)}.
    \end{aligned}
\end{equation}
Therefore, it suffices to derive nontrivial  lower bounds on $\prod_{i = 0}^{k-1} p^{(i)}$ and $\prod_{i = 0}^{k-1} q^{(i)}$. A lower bound on $q^{(i)}$ can be obtained by considering the conditions under which the largest marginal gain of  elements in $\R^{(i)}\cap \S^\star$ exceeds that in $\R^{(i)}\backslash \S^\star$ for all $i= 0,\dots,k-1$, i.e., 
\begin{equation}\label{eq:sofcon}
    \max_{j\in \R^{(i)}\backslash \S^\star} f_j(\S^{(i)}) < \max_{j\in \R^{(i)}\cap \S^\star} f_j(\S^{(i)}),
\end{equation}
with high probability. Doing so requires problem dependent information and a general result cannot be derived. In the supplementary we specify such a condition for the problem of sparse support selection. However, in Theorem \ref{thm:plus} below, we provide a lower bound on  $\prod_{i=0}^{k-1} p^{(i)}$.
\begin{theorem}\label{thm:plus}
Suppose the optimal solution to \eqref{eq:fmax} is unique. Let $r_i = \min(\frac{m}{k-i}\log\frac{1}{\epsilon},m)$, for all $i = 0,\dots,k-1$. Then, 
\begin{equation}\label{eq:bw}
    p = \prod_{i = 0}^{k-1} p^{(i)} \geq  \left(1-\epsilon\right)^{k-\log\frac{1}{\epsilon}}.
\end{equation}
Furthermore, if $\epsilon = \frac{1}{k^\alpha}$ for some $\alpha>1$, then $\limsup_{m,k\rightarrow \infty} p = 1$.
\end{theorem}
\begin{proof}
Note that since $r_i = m$ for all $i\geq k-\log\frac{1}{\epsilon}$, it follows that $p^{(i)} = 1$. Let us first consider the setting of sampling with replacement. There, it holds that
\begin{equation}
\begin{aligned}
\prod_{i = 0}^{k-1} p^{(i)}&=\prod_{i=0}^{k-\log\frac{1}{\epsilon}-1}\left(1-\left(1-\frac{k-i}{m}\right)^{r_i}\right)\\
&\geq \prod_{i=0}^{k-\log\frac{1}{\epsilon}-1} \left(1-\exp\left(-r_i\frac{k-i}{m}\right)\right)\\
&= \left(1-\epsilon\right)^{k-\log\frac{1}{\epsilon}}.
\end{aligned}
\end{equation}

Next, we consider the setting of sampling without replacement. For every $i<k-\log\frac{1}{\epsilon}$,
\begin{equation}
\begin{aligned}
p^{(i)} &= 1-\prod_{l = 0}^{r_i-1}\left(1-\frac{k-i}{m-l}\right) \\
&\stackrel{(a)}{\geq} 1-\left(1-\frac{k-i}{r_i}\sum_{l=0}^{r_i-1}\frac{1}{m-l}\right)^{r_i}\geq 1-\left(1-\frac{k-i}{m}\right)^{r_i}\\
&\stackrel{(b)}{\geq} 1-\exp\left(-r_i\frac{k-i}{m}\right)= 1-\epsilon,
\end{aligned}
\end{equation}
where $(a)$ is obtained by using the inequality relating arithmetic and geometric means, and $(b)$ is due to the fact that $(1+x)^y\leq e^{xy}$ for any real number $y\geq 1$. Therefore, just as in the case of sampling with replacement, \eqref{eq:bw} holds.

Finally, to establish the asymptotic result, using $(1+x)^y\leq e^{xy}$ we can show that
\begin{equation}
    \left(1-\epsilon\right)^{k-\log\frac{1}{\epsilon}} \geq 1-k\epsilon.
\end{equation}
Hence, if $\epsilon = \frac{1}{k^\alpha}$ for some $\alpha>1$, $\limsup_{m,k\rightarrow \infty} p = 1$ by the squeeze theorem. 
\end{proof}
Note that, as we argued in Corollary \ref{thm:fail2}, for the specific value of $\epsilon= \frac{1}{k^\alpha}$ the success probability of Algorithm \ref{alg:sgp} with  $r_i = r = \frac{m}{k}\log \frac{1}{\epsilon}$ goes to zero.  Additionally, it turns out that in the regime where $\epsilon \leq \frac{1}{k^\alpha}$, the proposed procedure requires $\O(m\log^2k) = \tilde{\O}(m)$ function evaluations as opposed to $\O(m\log k)$ function evaluations of \textsc{Stochastic-Greedy}.  

Since we established a lower bound on $\prod_{i = 0}^{k-1} p^{(i)}$ in \eqref{eq:bw}, it just remains to derive a nontrivial lower bound on $q^{(i)}$ in order to show existence of a sufficient condition for the exact identification of $\S^\star$, and establish a lower bound on the probability of success of Algorithm \ref{alg:sgp} with the proposed sampling strategy. While this general result requires problem-dependent information, in the Supplementary we build upon the idea of \eqref{eq:sofcon} to derive a lower bound on the probability of success for the problem of sparse support selection.
\subsection{Verifying the Theory}
\begin{figure*}[t]
	\centering
		\begin{subfigure}{.33\textwidth}
		\centering
		\includegraphics[width=\linewidth]{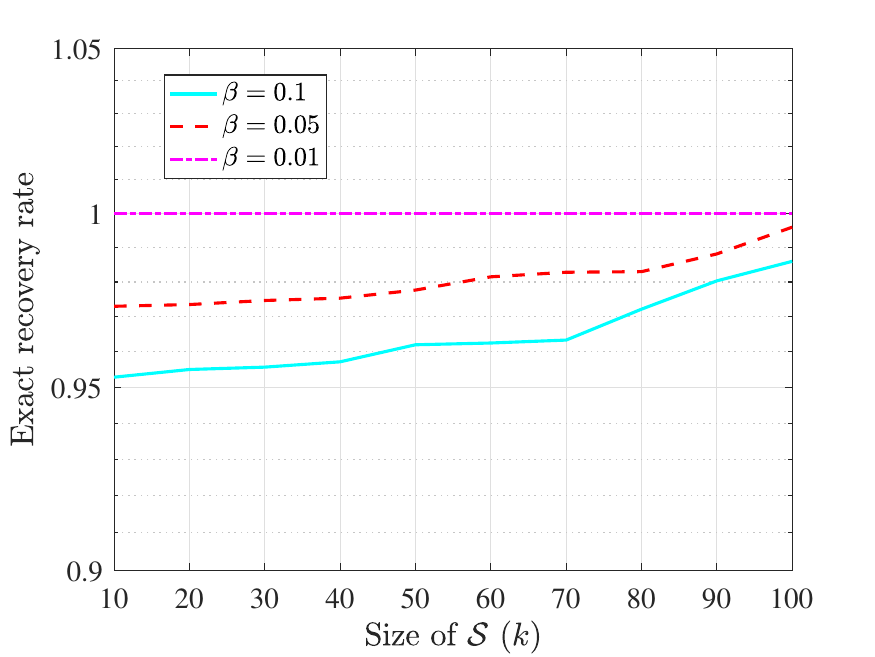}
		\caption{\footnotesize \textsc{Psg}'s recovery performance}
		\label{fig:sub3}
	\end{subfigure}
	\begin{subfigure}{0.33\textwidth}
		\centering
		\includegraphics[width=\linewidth]{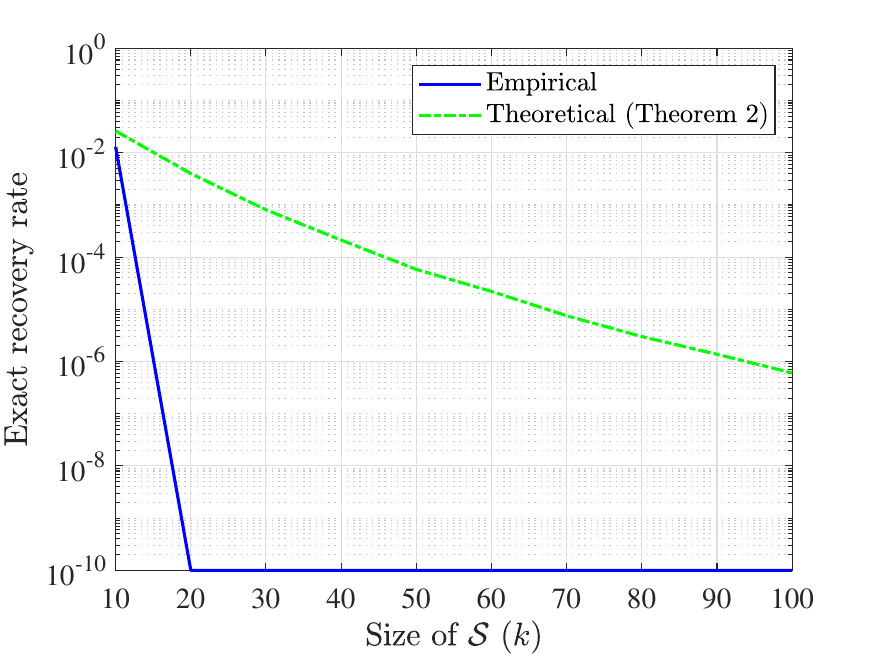}
		\caption{\footnotesize $r = m/\sqrt{k}$}
		\label{fig:sub1}
	\end{subfigure}%
	\begin{subfigure}{.33\textwidth}
		\centering
		\includegraphics[width=\linewidth]{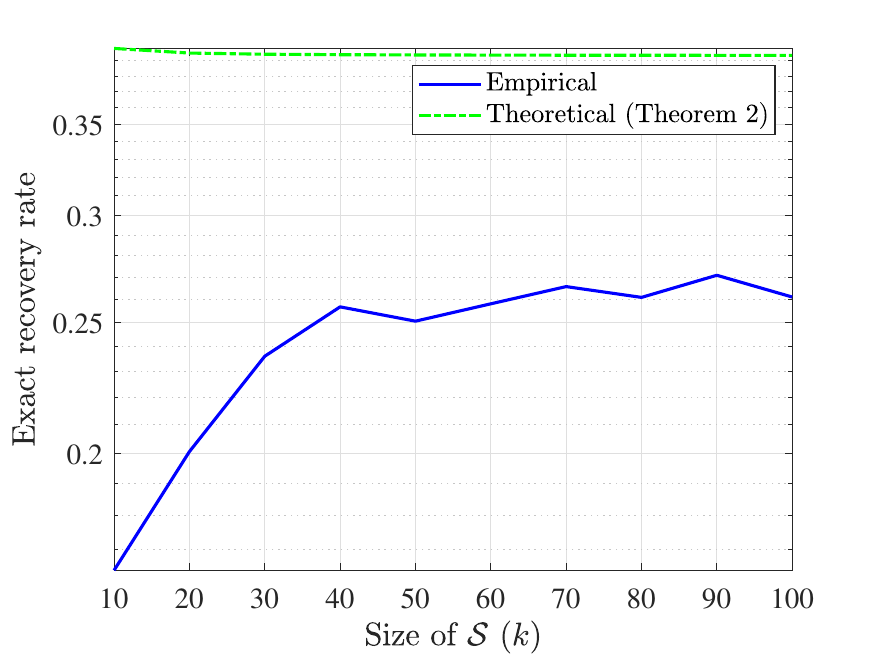}
		\caption{\footnotesize $r = m/2$}
		\label{fig:sub2}
	\end{subfigure}
	\caption{Empirical evaluation of the theoretical bounds established by Theorem \ref{thm:fail} and Corollary \ref{thm:fail2}.}
	\label{fig:sg}
	\vspace{-0.2cm}
\end{figure*}
In this section, we verify our theoretical results by comparing them to the empirical ones obtained via Monte Carlo (MC) 
simulations. Specifically, we consider the task of sparse support selection \cite{tropp2004greed,candes2005decoding}, in which we are given a linear measurement vector $\y=\A\x + {\bf \e}$ where $\x \in\Rb^{m}$ is a $k$-sparse unknown vector, i.e., a vector with at most $k$ non-zero components, $\y \in\Rb^{n}$ denotes the vector of measurements, $\A \in\Rb^{n \times m}$ is the coefficient matrix assumed to be full rank, and ${\bf \e} \in\R^{n}$ denotes the additive measurement noise vector.
The search for a sparse approximation of $\x$ leads to the NP-hard cardinality-constrained least-squares problem
\begin{equation}  \label{sls}
\begin{aligned}
& \underset{\x}{\text{minimize}}
\quad\|\y-\A\x\|^{2}_{2}
\qquad \text{subject to} \quad
\|\x\|_{0} \leq k,
\end{aligned}
\end{equation}
which can be interpreted as an instance of \eqref{eq:fmax} \cite{das2011submodular,khanna2017scalable}.

We consider a setting with increasing support size $k$ (varied from 
$10$ to $100$) and set the dimensions of the signal and the number of measurements to $m = 2k^{1.5}$ and $n = 6k\log(m/k\sqrt[6]{4\beta})$, 
respectively, for three different values of $\beta = 0.1, 0.05, 0.01$. In each trial, we select locations of the nonzero elements of $\x$ uniformly at random and draw those elements 
from a normal distribution. Entries of the coefficient matrix $\A$ are also generated randomly from ${\cal N}(0,\frac{1}{n})$. 
The results are averaged over 1000 Monte Carlo trials. Note that, as we show in the supplementary, in the above settings \textsc{Psg} is able to recover $\x$ exactly.

First, we investigate the exact performance of \textsc{Psg} with the schedule
\begin{equation}
    r_i=\frac{m}{k-i}\log\frac{1}{\epsilon}, \quad \epsilon = \frac{\beta}{k}
\end{equation}
for $\beta = 0.1, 0.05, 0.01$, and show the results in 
Fig.~\ref{fig:sg}(a). As can be seen there, the empirical exact recovery rate of \textsc{Psg} is very close to $1$; this coincides with the theoretical lower bound of $1-2\beta$ established in the supplementary that builds upon the insights of Theorem \ref{thm:fail} (i.e, the achieved rate is $0.8,0.9,0.98$ for $\beta = 0.1, 0.05, 0.01$, respectively). 

Next, we empirically verify the results of Theorem \ref{thm:fail} and Corollary \ref{thm:fail2} wherein we established an upper bound on the success probability 
of a variant of Algorithm \ref{alg:sgp}, named \textsc{Alg}, with a restricted uniform search space. Fig.~\ref{fig:sg} (b,c) compares this theoretical 
result with the empirical success rate for $r = m/\sqrt{k}$ and $r = m/2$, which correspond to instances of the two settings considered in Corollary \ref{thm:fail2}. Fig.~\ref{fig:sg}(b) shows that for $r = m/\sqrt{k}$ the success rate goes to zero as $k$ increases, as predicted by the first part of Corollary \ref{thm:fail2}.\footnote{Note that, in this setting, for $k\geq 20$ \textsc{Alg} failed in all of the trials; however, for  illustration purposes (i.e., to be able to show the plot in the logarithmic scale) 
we set the success 
rate of \textsc{Alg} for $k\geq 20$ to $10^{-10}$.} In Fig.~\ref{fig:sg}(c) we see that the success rate does not go to zero for $r = m/2$; moreover, it is always bounded by $1-e^{-0.5}\approx 0.39$, as claimed by the second part of Corollary \ref{thm:fail2}.

%% file: Secalg.tex
As argued in Section \ref{sec:fail}, finding a lower bound on the success probability requires problem-dependent information and thus a general result cannot be derived. Therefore, in this section we focus on analyzing the performance of \textsc{Psg} from the perspective of establishing nontrivial worst-case 
approximation factors.  
\subsection{Expected Approximation Factor}
We first establish bounds on the expected approximation factors of \textsc{Psg}. Then, assuming martingale structure of the marginal gains encountered in Algorithm \ref{alg:sgp}, we establish a high probability result for the approximation factor and present it in the next subsection.
\begin{proposition}\label{thm:expf}
Let $\S_{psg}$ denote the random subset selected by \textsc{Psg}, and let $c_f$ and $e_f$ be the multiplicative and additive weak submodularity constants of the 
submodularity ratio of the set function objective in \eqref{eq:fmax}, respectively. Then
\begin{equation}\label{eq:expbound}
	\E[f(\S_{psg})] \geq A_c f(\S^\star),\; \E[f(\S_{psg})] \geq A_e (f(\S^\star)-(k-1)e_f),
\end{equation}
where
\begin{equation}
\begin{aligned}
    A_c &:= 1-\prod_{i=1}^{k-\log \frac{1}{\epsilon}}\left(1-\frac{1-\epsilon^{\frac{k}{k-i-1}}}{kc}\right) \left(1-\frac{1}{kc}\right)^{\log \frac{1}{\epsilon}}, 
    \end{aligned}
\end{equation}
and
\begin{equation}
\begin{aligned}
    A_e &:= 1-\prod_{i=1}^{k-\log \frac{1}{\epsilon}}\left(1-\frac{1-\epsilon^{\frac{k}{k-i-1}}}{k}\right)\left(1-\frac{1}{k}\right)^{\log \frac{1}{\epsilon}}.
    \end{aligned}
\end{equation}
\end{proposition}
\begin{proof}
We first prove the result under a bounded $c_f$ assumption. To simplify the notation, let $f_{i}:= f(\S_{psg}^{(i)})$ and $f^\star:= f(\S^\star)$.

Since $c_f\geq 1$, it holds that $\frac{1}{r}(1+(r-1)c_f)\leq c_f$, for all $r$. Employing the first part of Lemma \ref{lem:curv} with $\S=\S_{psg}^{(i)}$ 
and $\T=\S^\ast\cup \S_{psg}^{(i)}$, and recalling monotonicity of $f$, we obtain
\begin{equation}\label{eq:1}
\begin{aligned}
f(\S^\ast)-f(\S_{psg}^{(i)})&\leq f(\S^\ast\cup \S_{psg}^{(i)})-f(\S_{psg}^{(i)})
\\&\leq  c_f\sum_{j\in \S^\ast\backslash \S_{psg}^{(i)}}f_j(\S_{psg}^{(i)}).
\end{aligned}
\end{equation}
We now proceed with the analysis by considering two distinct scenarios: the case where $r_i<m$ and hence the search space is strictly less than that of \textsc{Greedy}, and the case where $r_i = m$. Note that the first iteration in which the latter condition holds is $i_s := k-\log\frac{1}{\epsilon}$.

For any $i<i_s$, we need to find the probability that the current search space $\R^{(i)}$ samples an element from the optimal subset. In light of the analysis in Section \ref{sec:fail},
this probability for both sampling with and without replacement may be bounded below by 
\begin{equation}
\begin{aligned}
    \Pr\left(\R^{(i)} \cap (\S^\star\backslash\S^{(i)})\neq \emptyset\right) &\geq 1-e^{-\frac{r_i}{m}|\S^\star\backslash\S^{(i)}|} \\&\geq (1-\epsilon^{\frac{k}{k-i}})(\frac{|\S^\star\backslash\S^{(i)}|}{k}).
    \end{aligned}
\end{equation}
Given that for any $i<i_s$ \textsc{Psg} selects an element from $\R^{(i)}$ greedily, conditioned on $\R^{(i)} \cap (\S^\star\backslash\S^{(i)})\neq \emptyset$, the marginal gain of the selected value is larger than that of any element selected from $\R^{(i)} \cap (\S^\star\backslash\S^{(i)})$ uniformly at random (an in turn from $\S^\star\backslash\S^{(i)}$, given that the sampling is uniform). Therefore,
\begin{equation}
    \E\left[f_{(i+1)}(\S_{psg}^{(i)})|\S_{psg}^{(i)}\right]\geq \frac{1-\epsilon^{\frac{k}{k-i}}}{|\S^\star\backslash\S^{(i)}|} \sum_{j\in \S^\ast\backslash \S_{psg}^{(i)}}f_j(\S_{psg}^{(i)}).
\end{equation}
Since $|\S^\star\backslash\S^{(i)}|\leq k$, using \eqref{eq:1} and taking the total expectation we obtain for all $i<i_s$
\begin{equation}
\begin{aligned}
\E\left[(f_{i+1}-f_i)\right]\geq \frac{\left(1-\epsilon^{\frac{k}{k-i}}\right)}{kc_f}\left(f^\ast-\E\left[f_i\right]\right).
\end{aligned}
\end{equation}
For $i\geq i_s$, the search space of \textsc{Psg} is equivalent to that of \textsc{Greedy}. Therefore,
\begin{equation}
\begin{aligned}
f_{i+1}-f_i\geq \frac{1}{kc_f}\left(f^\ast-\E\left[f_i\right]\right).
\end{aligned}
\end{equation}
Let $\Delta_i = f^\ast-\E\left[f_i\right]$. Using induction,
\begin{equation}
    \Delta_{k-1} \leq (1-A_c)\Delta_0.
\end{equation}
Rearranging and noting that $\Delta_0 = f^\ast$ since $f(\emptyset) =0$ establishes the stated result.

The proof of the second part is analogous, except we leverage the second part of Lemma \ref{lem:curv} and define $\Delta_i = f^\ast-\E\left[f_i\right]-(k-1)e_f$ to arrive at the recursion $\Delta_{k-1} \leq (1-A_e)\Delta_0$,
thereby completing the proof.
\end{proof}
Note that, as expected, if the function $f$ is submodular, i.e., $c_f = 1$ and $e_f = 0$, then $A_c = A_e$ and the expected approximation factors in \eqref{eq:expbound} become identical.

It can be show using elementary algebra that $A_c \geq 1-e^{-\frac{1}{c_f}}-\frac{\epsilon}{c_f}$ and $A_e \geq 1-\frac{1}{e}-\epsilon$. Such bounds are expected as the search space of \textsc{Psg}, i.e. $r_i$, can be lower bounded by that of the \textsc{Stochastic-Greedy}, $r=\frac{m}{k}\log{(1/\epsilon)}$. Therefore, the expected approximation factors of \textsc{Psg} are at least as large as those of \textsc{Stochastic-Greedy}. However, it is worth noting that these lower bound are typically loose. To show this, we plot $A_e$ and compare the results to the approximation factors of \textsc{Greedy} and \textsc{Stochastic-Greedy} in Fig.~\ref{fig:compexpf} for two values of $\epsilon = 0.5,\;0.9$. As the figure demonstrates, there is a large gap between the expected approximation factors of \textsc{Psg} and \textsc{Stochastic-Greedy}. Note that the improvement in the approximation factor only requires a marginal increase in the number of oracle calls [cf. \eqref{eq:cmp}].
	\begin{figure}[t]
	\centering
		\includegraphics[width=0.4\textwidth]{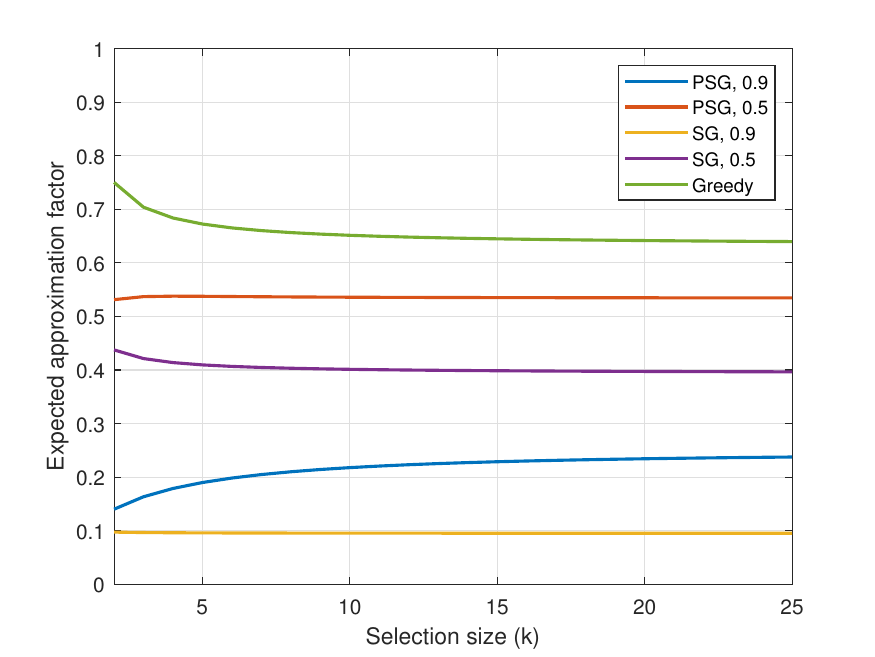}
		\caption{\footnotesize  Comparison of expected approximation factors of \textsc{Psg}, \textsc{Greedy}, and \textsc{Stochastic-Greedy} for $\epsilon = 0.5,\;0.9$ and different values of $k$.}
		\label{fig:compexpf}
	\end{figure}
\subsection{High-Probability Bounds for the Approximation Factor}
In Proposition \ref{thm:expf}, we established bounds on the expected approximation factor of the proposed scheme. Given that these results hold only on expectation, it is of interest to explore whether similar performance guarantees hold with high probability. The prior work \cite{hashemi2020randomized} has shown a high probability bound for Algorithm \ref{alg:sgp} when $r_i = r  = \frac{m}{k}\log \frac{1}{\epsilon}$. To derive this result, each step of Algorithm \ref{alg:sgp} is interpreted as an approximation of the marginal gains of
the selected elements using \textsc{Greedy},
\begin{equation}\label{eq:eta}
    f_{j_{(alg)}}(\S^{(i)}) = \eta^{(i)}  f_{j_{(g)}}(\S^{(i)}),
\end{equation}
where subscripts $alg$ and $g$ refer to the elements selected by Algorithm \ref{alg:sgp} and \textsc{Greedy}, respectively, and $\{\eta^{(i)}\}_{i=1}^k$ is a collection of random variables in $(0,1]$. 

The result of \cite{hashemi2020randomized}, however, relies on the assumption that $\{\eta^{(i)}\}_{i=1}^k$ are independent, which is not realistic since the element selected in each iteration depend on the selected subset so far. In Proposition \ref{thm:hpf}, we relax that independence assumption and show that as long as $\{\eta^{(i)}\}_{i=1}^k$ form a martingale, one can derive a high-probability performance guarantee. Additionally, we derive high-probability bounds using both additive and multiplicative weak submodularity constants.
\begin{proposition}\label{thm:hpf}
    Instate the notation and assumptions of Proposition \ref{thm:expf}. Let $\eta$ be a martingale with $\E(\eta) \geq \mu$ satisfying the conditions of Theorem \ref{thm:mar}. Then, it holds with probability of at least $1-\delta$ (for any $\delta>0$) that
	\begin{equation}\label{eq:hpbound}
	f(\S_{psg}) \geq A^{hp}_c f(\S^\star),\; f(\S_{psg}) \geq A^{hp}_e (f(\S^\star)-(k-1)e_f),
	\end{equation}
\end{proposition}
where
\begin{equation}
\begin{aligned}
    A^{hp}_c &:= 1-\exp\left(-\frac{1}{c_f}\left(\mu +\sqrt{\frac{\log \frac{1}{\delta}}{2k}}+\frac{\log \frac{1}{\epsilon}}{k}\right)\right),
    \end{aligned}
\end{equation}
and
\begin{equation}
\begin{aligned}
    A^{hp}_e &:= 1-\exp\left(-\left(\mu +\sqrt{\frac{\log \frac{1}{\delta}}{2k}}+\frac{\log \frac{1}{\epsilon}}{k}\right)\right).
    \end{aligned}
\end{equation}
\begin{proof}
Let $f_{i}:= f(\S_{psg}^{(i)})$ and $f^\star:= f(\S^\star)$. Note that for any element $j$ it holds that
\begin{equation}
   f_{j}(\S_{psg}^{(i)}) \leq f_{(j)_{psg}}(\S_{psg}^{(i)})\leq  f_{(j)_{g}}(\S_{psg}^{(i)}).
\end{equation}
Combining the above result with \eqref{eq:1} and \eqref{eq:eta} we obtain
\begin{equation}
\begin{aligned}
f_{i+1}-f_i\geq \frac{\eta^{(i+1)}\mathbb{I}(i<i_s)+\mathbb{I}(i\geq i_s)}{kc_f}\left(f^\ast-f_i\right).
\end{aligned}
\end{equation}
Using an inductive argument similar to the one in the proof of Proposition \ref{thm:expf}, we obtain
\begin{equation}\label{eq:prop2}
\begin{aligned}
f_{k-1}\geq f^\ast \left[1-e^{-\frac{1}{kc_f}\sum_{i=1}^{i_s} \eta^{(i) }}\left(1-\frac{1}{kc_f}\right)^{\log \frac{1}{\epsilon}}\right].
\end{aligned}
\end{equation}
Next, we use the result of Theorem \ref{thm:mar} to bound the deviation of $\sum_{i=1}^{i_s} \eta^{(i) }$ from its mean with high probability. By Theorem \ref{thm:mar} and the assumption that $ \eta^{(i) } \in (0,1]$,
\begin{equation}
 \Pr\left(\sum_{i=1}^{i_s} \eta^{(i) }-i_s \mu>\lambda\right) \leq\exp\left(-\frac{\lambda^2}{2 i_s}\right).
\end{equation}
Setting the right-hand side of the above iinequality to $\delta$ yields that with probability exceeding $1-\delta$
\begin{equation}
    \sum_{i=1}^{i_s}\eta^{(i) } \leq i_s \mu +\sqrt{\frac{i_s}{2}\log \frac{1}{\delta}} \leq k \mu +\sqrt{\frac{k}{2}\log \frac{1}{\delta}},
\end{equation}
where we relied on the fact that $i_s \leq k$. Using this bound in \eqref{eq:prop2} along with $(1+x)^y\leq e^{xy}$ for any real number $y\geq 1$ proves the first part of the proposition.

Leveraging the second part of Lemma \ref{lem:curv} and repeating similar arguments furnishes the second result.
\end{proof}
Proposition \ref{thm:hpf} establishes high probability bounds on the worst-case performance of the proposed scheme. For large $k$, the approximation factors may be approximated by $1-e^{-\mu/c_f}$ and $1-e^{-\mu}$ for multiplicative and additive weak submodularity constants, respectively. Therefore, if $\mu$ is near $1$, as expected, the marginal gains of elements selected by \textsc{Psg} are relatively close to those that would have been selected by \textsc{Greedy}, which in turn means we recover the guarantees of \textsc{Greedy} in that regime.

%% file: Secsim.tex
We demonstrate efficacy of the proposed algorithm in two applications, namely column subset selection for sparse subspace clustering and observation selection for target tracking.
\subsection{Column Subset Selection for Subspace Clustering}
\begin{figure*}[t]
	\centering
	\minipage[t]{1\linewidth}
	\begin{subfigure}[t]{.33\linewidth}
		\includegraphics[width=\textwidth]{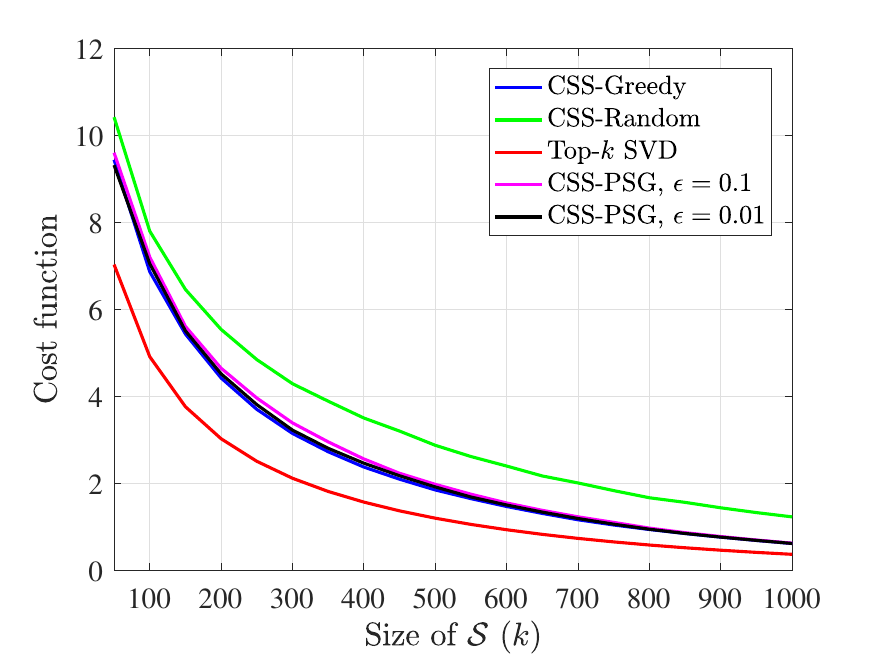}
		\caption{\footnotesize  Reconstruction error}
	\end{subfigure}
	\begin{subfigure}[t]{.33\linewidth}
		\includegraphics[width=\linewidth]{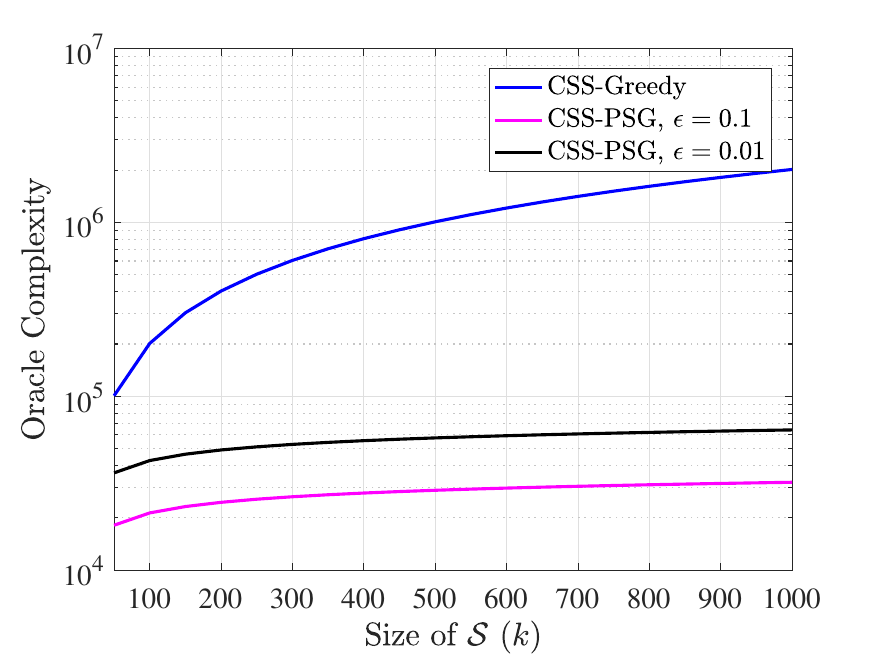}
		\caption{\footnotesize Oracle complexity}
	\end{subfigure}
	\begin{subfigure}[t]{.33\linewidth}
		\includegraphics[width=\linewidth]{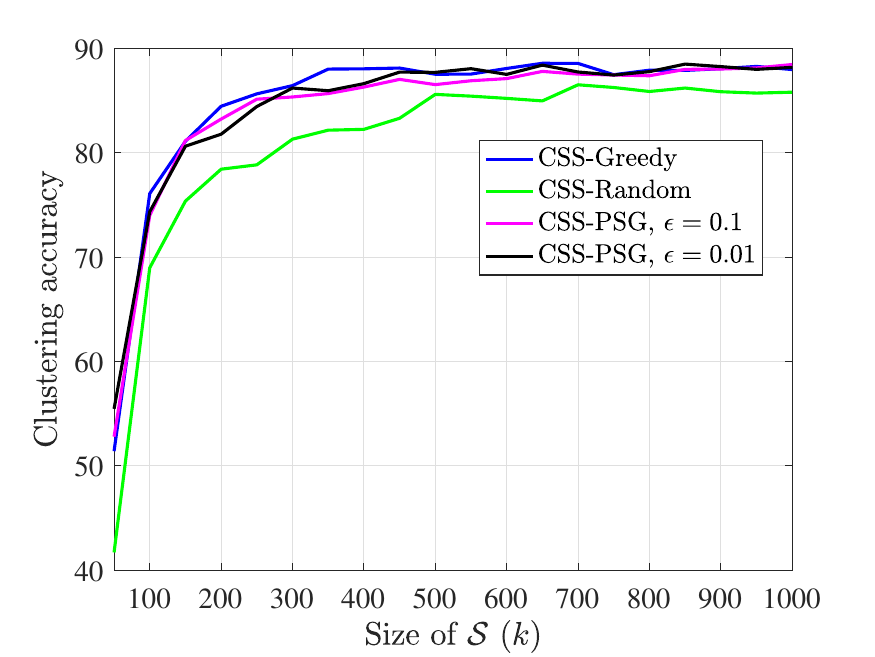}
		\caption{\footnotesize Clustering accuracy}
	\end{subfigure}
	\caption{Performance comparison of SSC with various CSS-based dimensionality reduction schemes on EYaleB
		dataset consisting of face images under 64 different illumination conditions.}
	\label{fig:real}
	\endminipage 
	\vspace{-0.4cm}
\end{figure*}
Here we present results of an empirical evaluation of the proposed \textsc{Psg} scheme; specifically, the performance of \textsc{Psg} is compared to several baselines on the task of dimensionality reduction via column subset selection (CSS) \cite{tropp2009column} for sparse subspace clustering (SSC) \cite{elhamifar2009sparse,you2015sparse,hashemi2018accelerated,hashemi2018evolutionary}.

The goal of CSS is to identify a subset $\S$, $|\S| = k$, of the set of $m$ columns of a data matrix $\A \in \mathbb{R}^{n \times m}$ that best approximate the entire data matrix. Formally, the task of identifying $\S$ can be cast as the optimization problem
\begin{equation}\label{eq:css}
\begin{aligned}
& \underset{\S}{\text{minimize}}
\quad \|\A-\P_{\S}\A\|_F^2
& \text{s.t.}\quad  |\S| = k,
\end{aligned}
\end{equation}
where $\P_\S = \A_S \A_S^\dagger$ is the projection operator 
(specifically, projection onto the span of columns of $\A_\S$) and $\A_\S^\dagger=\left(\A_\S^{\top}\A_\S\right)^{-1}\A_\S^{\top}$ 
denotes the Moore-Penrose pseudo-inverse of $\A_\S$. Since $\A = \P_{\S}\A+ (\I-\P_{\S})\A$ and $\|\A\|_F^2 = \|\P_{\S}\A\|_F^2+ \|(I-\P_{\S})\A\|_F^2$ by properties of projection matrices, \eqref{eq:css} can equivalently be written as an instance of the weak submodular maximization task \eqref{eq:fmax} \cite{farahat2015greedy,bhaskara2016greedy}.

We aim to use Algorithm \ref{alg:sgp} as a CSS-based dimensionality reduction technique to reduce the cost of performing clustering via SSC \cite{elhamifar2009sparse,you2015sparse}. That is, using a lower dimensional data matrix $\A_{\S_g}$ obtained via CSS, we learn the representation matrix $\mathbf{C}$ by solving
\begin{equation}
\underset{\mathbf{C}}{\text{minimize}}
\quad \|\A_{\S_g}-\A_{\S_g}\mathbf{C}\|_F^2+\lambda \|\mathbf{C}\|_1,
\end{equation}
and then employ spectral clustering \cite{ng2002spectral} on $\mathbf{W} = |\mathbf{C}|+|\mathbf{C}|^\top$ to segment the data points.

We consider the proposed \textsc{Psg} scheme with two values of $\epsilon$: 
$\epsilon = 0.1$ and $\epsilon = 0.01$. We consider \textsc{Greedy} and random column subset selection as the benchmarking schemes. Additionally, we use the best rank-$k$ approximation of a matrix (i.e., top-$k$ SVD) to serve as an upper bound on the achievable performance; note that this scheme explicitly minimizes the Forbenius reconstruction criterion. We compare performance of the above algorithms using the real EYaleB dataset \cite{georghiades2001few} which contains frontal face images
of 38 individuals under 64 different illumination conditions (see Fig.~\ref{fig:face}). There are $m=2414$ columns (i.e., features) in this dataset; we select $k$ out of $m=2414$ columns, where $k$ varies from 100 to 1000, and apply the SSC method of \cite{you2015sparse} to cluster the data points based on the selected features.

\begin{figure*}[t]
	\centering
	\minipage[t]{1\linewidth}
	\begin{subfigure}[t]{0.33\linewidth}
		\includegraphics[width=\textwidth]{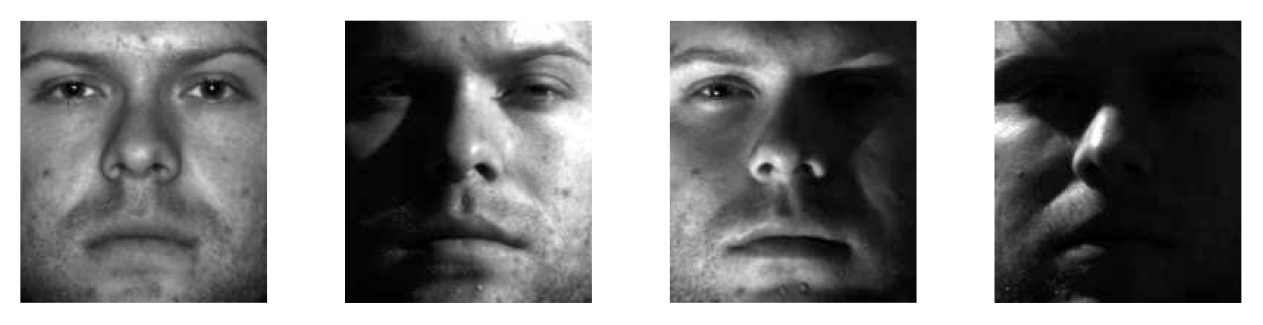}
	\end{subfigure}
	\begin{subfigure}[t]{.33\linewidth}
		\includegraphics[width=\linewidth]{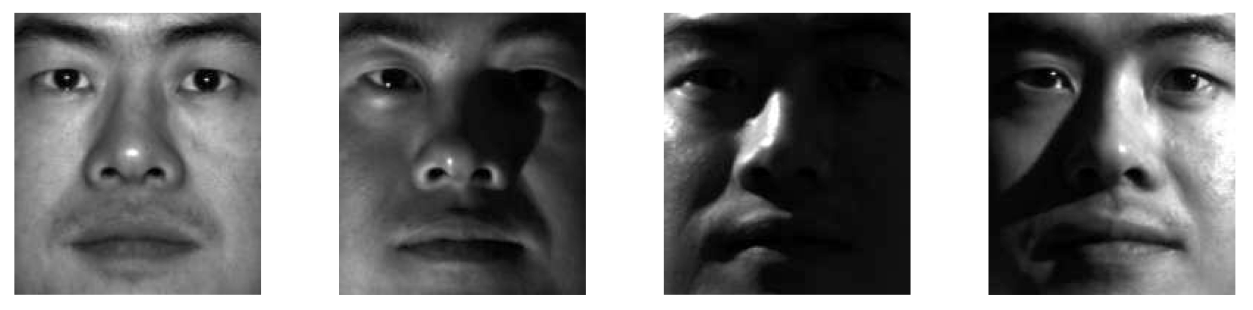}
	\end{subfigure}
	\begin{subfigure}[t]{.33\linewidth}
		\includegraphics[width=\linewidth]{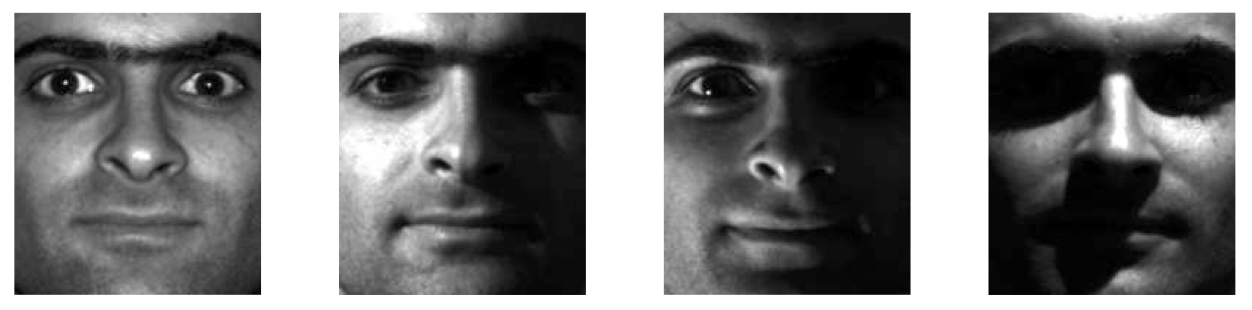}
	\end{subfigure}
	\caption{Face clustering: given images of multiple subjects, find images that belong to the same subject (examples from the EYaleB dataset \cite{georghiades2001few}).}
	\label{fig:face}
	\endminipage 
\end{figure*} 

Fig. \ref{fig:real} shows the performance of various column subset selection schemes as well as the top-$k$ SVD approach. In Fig. \ref{fig:real}(a) we observe that the reconstruction errors of \textsc{Greedy} and the proposed scheme are nearly identical, and that as we increase the number of selected columns the reconstruction error decreases; this is consistent with the fact that $f(\S)$ is a monotone function. Fig. \ref{fig:real}(b) shows a significant computational complexity improvement that the proposed scheme provides over the greedy CSS method. Since the complexity of Algorithm \ref{alg:sgp} increases logarithmically in $k$, the cost of selecting more columns is relatively small compared to the greedy approach. Note that we observe  $\epsilon = 0.1$ achieves the best tradeoff between computational costs and performance. Furthermore, depending on the amount of data available, the value of $\epsilon$ can be tuned using cross-validation. Finally, in Fig. \ref{fig:real}(c) we compare the clustering accuracy of SSC applied to a subset of features selected by different schemes. As the figure shows, clustering performance of SSC combined with the proposed CSS method is nearly identical to that of the conventional greedy approach; moreover, both achieve superior accuracy compared to schemes that randomly select subsets of columns.
\subsection{Observation Selection in Target Tracking}
	\begin{figure}[t]
	\centering
		\includegraphics[width=0.48\textwidth]{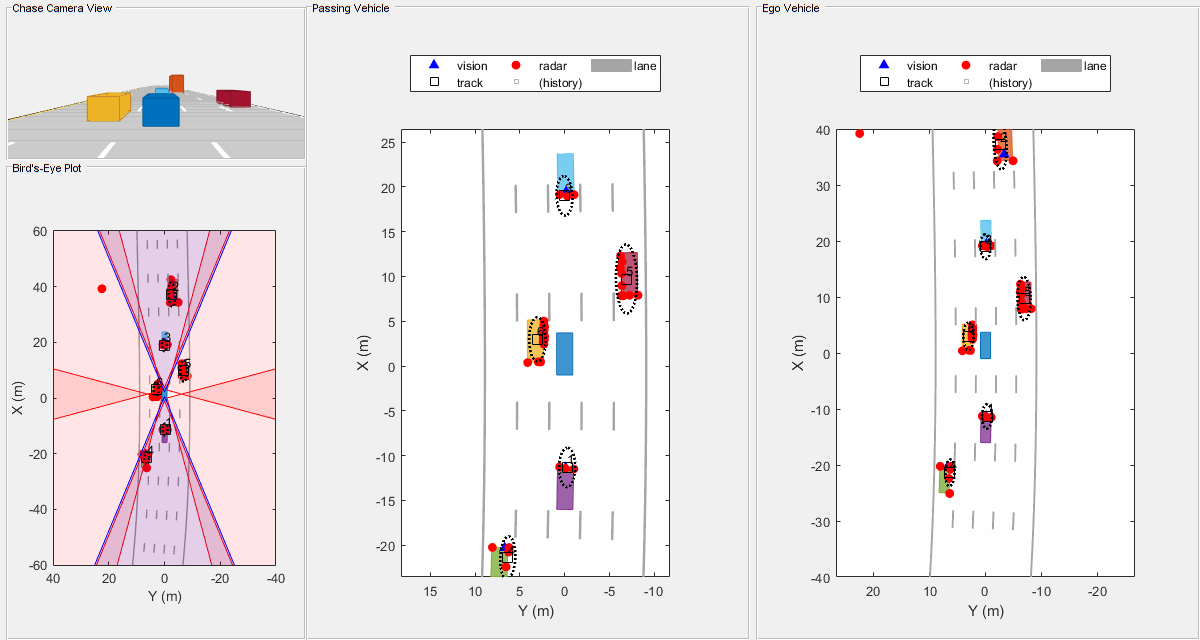}
		\caption{\footnotesize  The object tracking scenario: an ego vehicle (dark blue) moves on a highway with five lanes and aims to identify and track the vehicles in its surroundings using six radar and two camera sensors.}
		\label{fig:ego}
		\vspace{-3mm}
	\end{figure}
\begin{figure}[t]
	\centering
	\begin{subfigure}[t]{0.95\linewidth}
		\includegraphics[width=\linewidth]{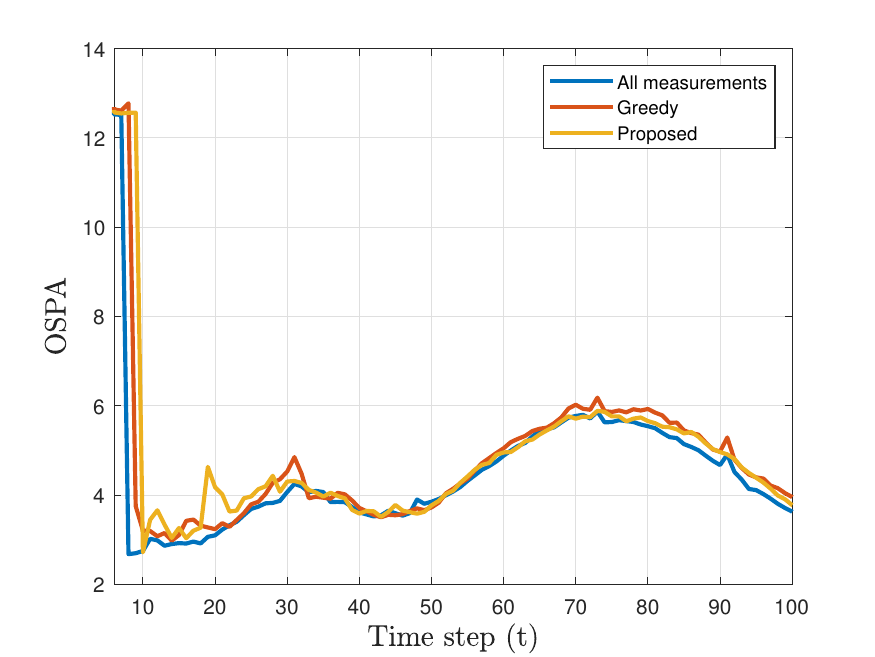}
		\caption{\footnotesize Tracking error}
	\end{subfigure}
	\begin{subfigure}[t]{0.95\linewidth}
		\includegraphics[width=\linewidth]{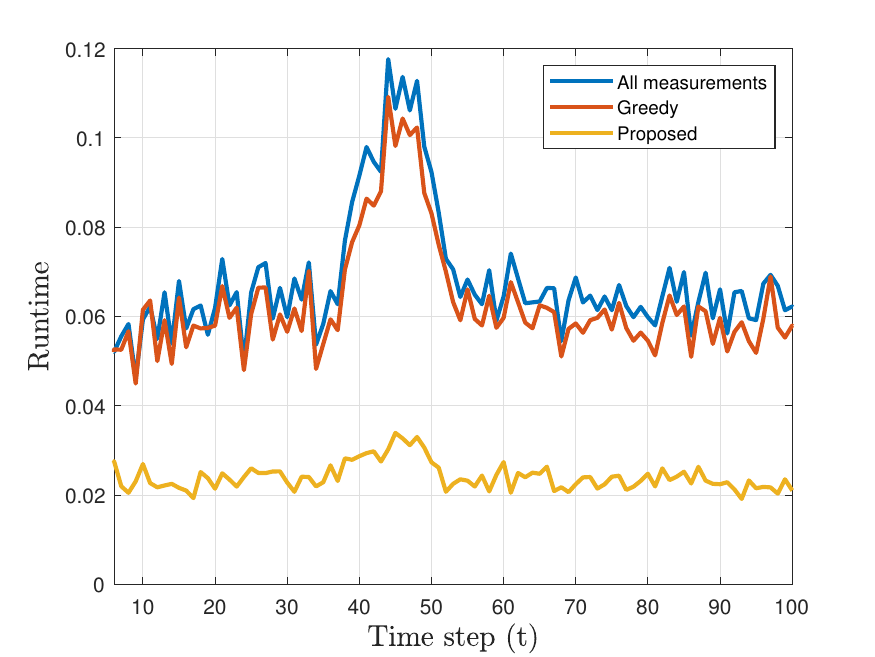}
		\caption{\footnotesize Runtime}
	\end{subfigure}
	\caption{Performance comparison of the proposed observation selection algorithm with the greedy approach, and an approach that utilizes all of the gathered observations. The proposed scheme results in nearly the same OSPA while significantly reducing the running time of tracking.}
	\label{fig:target}
	\vspace{-4mm}
\end{figure}
We consider a realistic multi-object tracking application by an autonomous car, referred to as the ego vehicle, that moves on a highway. Perception plays a major role in the planning subsystem of an autonomous vehicle; accurately identifying and tracking the critical objects in its surrounding, e.g., other vehicles on the highway, are critical in the design of the perception subsystem \cite{ghasemi2020task}. Given the typically large amount of data gathered by the vehicle's sensors,  observation selection techniques may be employed to reduce the computational burden of tracking algorithms and improve their runtime \cite{ghasemi2019online,hashemi2020randomized}. Motivated by this argument, we resort to performing observation selection via the proposed scheme.

We consider a scenario with an ego vehicle that moves on a stretch of 500 meters of a typical highway road with five lanes, along with six other vehicles (see Fig. \ref{fig:ego} and the supplementary gif file). All the vehicles move with the speed of $25\;km/s$, except for the passing vehicle which has a speed of  $35\;km/s$. The ego vehicle  has six radar sensors and
two vision sensors covering the 360 degrees field of view. The sensors have
some overlap and some coverage gap. The ego vehicle is equipped with a
long-range radar sensor and a vision sensor on both the front and the
back of the vehicle. Each side of the vehicle has two short-range radar
sensors, each covering 90 degrees. One sensor on each side covers the area from
the middle of the vehicle to the back. The other sensor on each side
covers the area from the middle of the vehicle forward. We set the properties of the sensors such that they are in agreement with the realistic characteristics of radar and camera sensors \cite{matlab}. The sensors gather noisy observations about the location and velocity of the surrounding vehicles. After selecting the observations via a subset selection algorithm, the selected observations are used in a GGIW-PHD extended object tracker that relies on a rectangular target model for the surrounding vehicles \cite{granstrom2011tracking,granstrom2012extended,granstrom2016extended}. We assess the performance of the observation selection and tracking algorithms based on the optimal sub-pattern assignment (OSPA) metric \cite{schuhmacher2008consistent,ristic2011metric} as well as the time it takes to produce the estimated locations.

Fig. \ref{fig:target} shows the performance comparison of the proposed scheme with greedy observation selection that selects 50\% of the observations. We also provide a comparison with a method that uses all of the gathered observations in the tracking algorithm. As Fig. \ref{fig:target} (a) shows, all schemes achieve similar tracking accuracy. Note that OSPA increases from time $t=50$ to $t=75$ as in this interval the passing vehicle gets farther from the ego vehicle, but it is still considered as a tracked object in the OSPA metric. However, after $t=75$, the passing vehicle is not in the line of sight of the ego car anymore, resulting in the OSPA metric to improve.

Fig. \ref{fig:target} (b) depicts the runtime comparison and demonstrates that using the proposed scheme to perform observation selection significantly reduces the cost of object tracking. Note that in the interval $(35,50)$, the passing vehicle is in the proximity of the ego vehicle. Therefore, the mounted sensors produce significantly more observations, thereby increasing the runtime of the benchmarking schemes. However, thanks to the progressive random sampling strategy, the runtime of the proposed scheme does not increase significantly in this interval.